\documentclass[%
 aip,
 amsmath,amssymb,
 reprint,%
]{revtex4-1}

\usepackage{graphicx}% Include figure files
\usepackage{dcolumn}% Align table columns on decimal point
\usepackage{bm}% bold math
\usepackage[utf8]{inputenc}
\usepackage[T1]{fontenc}
\usepackage{mathptmx}
\usepackage{amsmath,amsfonts,amssymb,amsthm,amscd}
\usepackage{graphicx,mathtools}
\usepackage{perpage}
\usepackage{url}
\usepackage{color}
\usepackage[T1]{fontenc}
\usepackage{subcaption}
\usepackage{mathrsfs}
\usepackage{booktabs}
\newtheorem{theorem}{Theorem}[section]
\newtheorem{lemma}[theorem]{Lemma}
\newtheorem{proposition}[theorem]{Proposition}

\newtheorem{remark}[theorem]{Remark}
\newtheorem{definition}[theorem]{Definition}

\DeclareMathOperator{\zero}{zero}

\newcommand{\R}{\mathbb{R}}

\renewcommand{\S}{\mathbb{S}}
\newcommand{\RR}{\mathbf{\mathscr{R}}}

\newcommand{\eee}{\mathrm{e}}
\newcommand{\ee}{\varepsilon}
\newcommand{\ddd}{\mathrm{d}}
\newcommand{\dd}{\delta}

\newcommand{\mi}{\mathrm{i}}
\renewcommand{\a}{\alpha}
\newcommand{\p}{\partial}
\newcommand{\I}{\textnormal{I}}

\begin{document}

\preprint{AIP/123-QED}

\title[Two-community noisy Kuramoto model with general interaction strengths: Part I]{Two-community noisy Kuramoto model with general interaction strengths: Part I}
% Force line breaks with \\

\author{S. Achterhof}
\affiliation{ 
Mathematical Institute, Leiden University, P.O.\ Box 9512,
2300 RA Leiden, The Netherlands.%\\This line break forced with \textbackslash\textbackslash
}%

\author{J. M. Meylahn}
 \email{j.m.meylahn@uva.nl}
\affiliation{%
Amsterdam Business School, University of Amsterdam, P.O.\ Box 15953,
1001 NL Amsterdam, The Netherlands.%\\This line break forced% with \\
}%

\date{\today}% It is always \today, today,

\begin{abstract}
We generalize the study of the noisy Kuramoto model, considered on a network of two interacting communities, to the case where the interaction strengths within and across communities are taken to be different in general. By developing a geometric interpretation of the self-consistency equations, we are able to separate the parameter space into ten regions in which we identify the maximum number of solutions in the steady state. Furthermore, we prove that in the steady-state only the angles $0$ and $\pi$ are possible between the average phases of the two communities and derive the solution boundary for the unsynchronized solution. Lastly, we identify the equivalence class relation in the parameter space corresponding to the symmetrically synchronized solution.  
\end{abstract}

\maketitle

\begin{quotation}
The Kuramoto model is a model for studying synchronization of oscillators (e.g. fireflies flashing). We consider two groups of oscillators, synchronizing within and across groups. Studying the stationary-states (the states after waiting a long time) of the system leads to a system of equations that cannot be solved analytically. We introduce a geometric interpretation of these equations that allows us to analyze when and how many solutions are possible given a vector of model parameters. It also allows us to identify when symmetric solutions (solutions where the two groups are equally synchronized) and unsychronized solutions occur.
\end{quotation}

\section{Background and motivation}
The Kuramoto model on a two-community network with general interaction strengths has recently received attention in the physics literature \cite{Abrams2008, Hong2011, Hong2012, Kotwal2017, Sonnenschein2015}. Here three approaches have been used, namely, the Ott-Antonsen Ansatz, the Gaussian approximation and the reduction to circular cumulants approach. These are all methods of approximating the low dimensional dynamics of the system, i.e., dynamics of the order parameters. The Gaussian approximation, developed in \cite{Sonnenschein2013, Hannay2018}, can be applied to the noisy Kuramoto model, while the reduction to circular cumulants developed in \cite{Goldobin2018, Tyulkina2018} can be applied to the noisy Kuramoto model only in the small noise limit.  

A common theme in the aforementioned literature is the appearance of bifurcation points, chimera states and traveling waves arising from this simple modification to the Kuramoto model. This suggests that slight increases in terms of complexity on the underlying interaction network structure results in a plethora of complex phenomena.

In a recent paper \cite{Meylahn2020}, one of the authors fully classified the phase diagram for the two-community noisy Kuramoto model in the case when the pair of interaction strengths in the two communities as well as the pair of interaction strengths between the two communities are taken to be the same (henceforth referred to as the symmetric case). This reduces the parameter space to two dimensions. In this case three types of solutions exist: the unsychronized solution, the symmetrically synchronized solution (when both communities are synchronized to the same degree) and the non-symmetrically synchronized solution (when on community is more synchronized than the other). The non-symmetrically synchronized solution appears as a pitchfork bifurcation, resulting in a bifurcation line (or solution boundary) in the phase diagram. The paper also proves that in this simplified case the phase difference between the average phases of the two communities must be zero or $\pi$, which significantly simplifies the analysis.  

The two-community noisy Kuramoto model is, however, not only of interest to mathematicians and physicists, but is also relevant for neurophysiologists. The Suprachiasmatic nucleus (SCN), or body-clock, is a network of neurons in the brain responsible for dictating all bodily rhythms and surprisingly has the same two-community structure in all mammals. The results on the symmetric case might explain the observation of a variety of interesting phenomena in experiments, for example, the transitions to a phase-split state of the SCN observed in mice and rats when exposed to constant light as shown in a recent paper by one of the authors \cite{Rohling2020}.

The presence or absence of a variety of compounds or chemicals in the SCN changes the strength of the interactions between neurons, and the concentration of these chemicals is in turn influenced by a variety of external/environmental factors. This means that an accurate model of the SCN would incorporate time-dependent interaction strength parameters. A first step in this direction is to generalize the results of the previous paper to the case where we allow for four different interaction strength parameters: two for the interactions within the communities and two for the interactions between the communities. This is the goal of this series of papers.

Another interesting application of the two-community Kuramoto model is to the dynamics of opinion formation. As argued in \cite{Binmore2004} political opinions should be represented in at least two dimensions. Furthermore, individuals tend to update their opinion based on the opinions of individuals they come into contact with. These interactions can be both negative and positive (rejecting an opinion due to previous disagreements with an individual or accepting an opinion due to previous agreement). This makes the Kuramoto model with positive and negative interactions a natural candidate to study the dynamics of opinion formation and especially the phenomenon of polarization. This has been explored to some degree by \cite{Hong2011a}, \cite{Xiao2019} and a modification of the Kuramoto model called the Opinion Changing Rate model is studied in \cite{Pluchino2005}.

The paper proceeds as follows. In Section \ref{sec:model} we define the model we will consider and derive the set of self-consistency equations determining all solutions of the model. In Section \ref{sec:overview} we introduce a geometric interpretation of the self-consistency equations, which allows us to split the parameter space into ten regions that can be analyzed separately. In Section \ref{sec:unsynch} we identify regions in the phase space in which the unsynchronized solution is the only solution and analyze a special solution in which both communities are equally synchronized (called the symmetrically synchronized solution).

\section{The model}
\label{sec:model}
Consider two populations of oscillators with sizes $N_1$ and  $N_2$ and with internal mean-filed interactions of strength $\frac{K_1}{N_1}$ and $\frac{K_2}{N_2}$. Furthermore the oscillators in community $1$ experience a mean-field interaction with the oscillators in community 2 of strength $\frac{L_1}{N_2}$ and the oscillators in community $2$ experience a mean-field interaction with the oscillators of community 1 of strength $\frac{L_2}{N_1}$ (see Figure \ref{fig:comicgraph}). We assume that $K_1, K_2 \in \R$ and $L_1, L_2 \in \R \setminus\{0\}$.

\begin{definition}[Two-community noisy Kuramoto model]
The evolution of $\theta_{1,i}$, $i = 1, \ldots, N_1$, on $\S = \R/ 2\pi$ is governed by the SDE
\begin{align}
\ddd \theta_{1,i}(t) &= \omega_{1,i} \ddd t + \frac{K_1}{N_1 + N_2} \sum_{k = 1}^{N_1} \sin( \theta_{1,k}(t) - \theta_{1,i}(t) ) \ddd t  \label{kura1} \\
 	&+ \frac{L_1}{N_1 + N_2} \sum_{l = 1}^{N_2} \sin(\theta_{2,l}(t) - \theta_{1,i}(t) ) \ddd t + \sqrt{D} \ddd W_{1,i}(t).\nonumber  
\end{align}

As initial condition we take $\theta_{1,i}(0)$ for $i = 1, \ldots, N_1,$ i.i.d. and drawn from a common probability distribution $\rho_1$ on $\S$. The natural frequencies $\omega_{1,i}$, $i = 1, \ldots, N_1$ are i.i.d. and are drawn from a common probability distribution $\mu_2$ on $\R$. 
\\

The phase angles of the oscillators in community 2 are denoted by $\theta_{2,j}$, $j = 1, \ldots, N_2$, and their evolution on $\S = \R/ 2\pi$ is governed by the SDE
\begin{align}
\ddd \theta_{2,j}(t) &= \omega_{2,j} \ddd t + \frac{K_2}{N_1 + N_2} \sum_{l = 1}^{N_2} \sin( \theta_{2,l}(t) - \theta_{2,j}(t) ) \ddd t \nonumber \\
 	&+ \frac{L_2}{N_1 + N_2} \sum_{k = 1}^{N_1} \sin(\theta_{1,k}(t) - \theta_{2,j}(t) ) \ddd t + \sqrt{D} \ddd W_{2,j}(t). \label{kura2}
\end{align}
As initial condition we take $\theta_{2,j}(0), j = 1, \ldots, N_2,$ are i.i.d. drawn from a common probability distribution $\rho_2$ on $\S$. The natural frequencies $\omega_{2,j}$, $j = 1, \ldots, N_1$ are i.i.d. and are drawn from a common probability distribution $\mu_2$ on $\R$. 
Furthermore $\left(W_{1,i}\right)_{t \geq 0}$, $i = 1,\ldots, N_1$ and $\left(W_{2,j}\right)_{t \geq 0}$, $j = 1,\ldots, N_2$ are two independent standard Brownian motions and we call $D > 0$ the noise strength. 
\end{definition}

\begin{figure*}[!ht]
\centering
\includegraphics[width=1.0\textwidth]{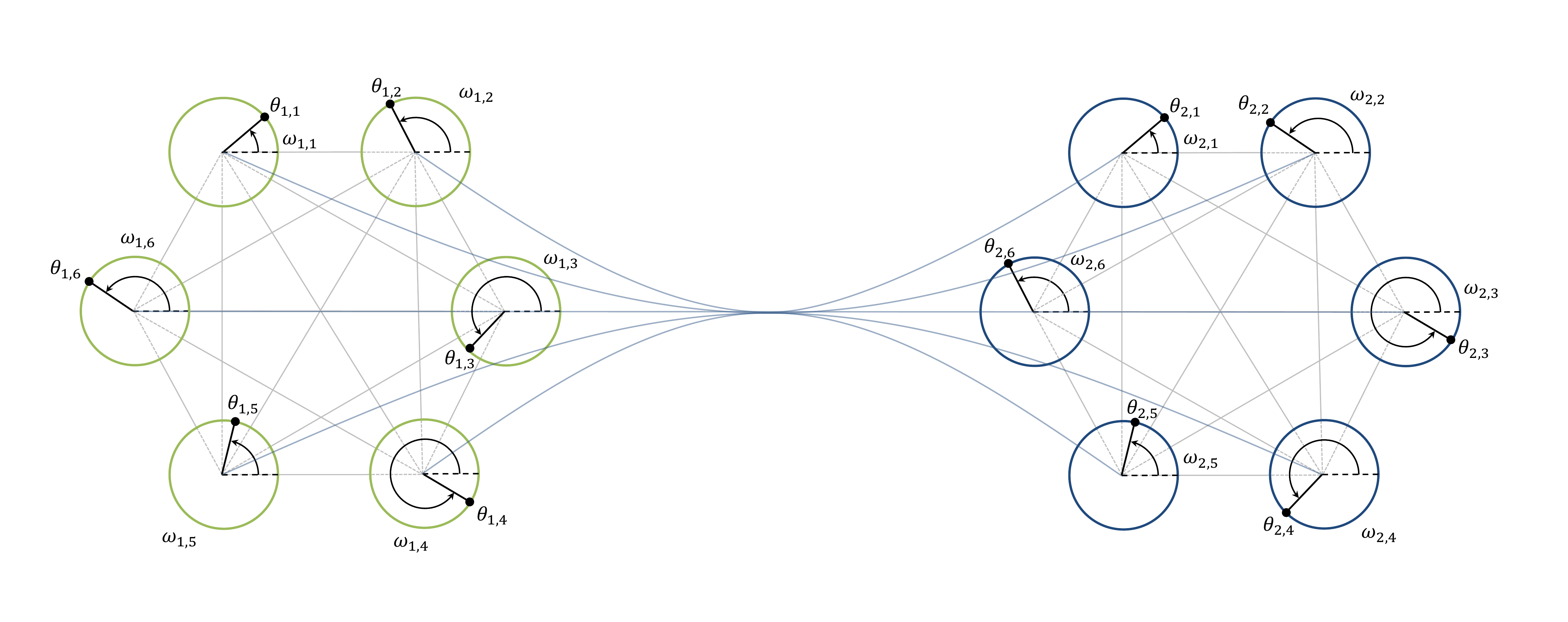}
\caption{Graphical representation of the two-community noisy Kuramoto model, with $N_1 = N_2 = 6$.}
\label{fig:comicgraph}
\end{figure*}

In order to monitor the dynamics in each community, let us define the \emph{order parameter of community 1 and community 2}, respectively:
\begin{align}
r_{1,N_1}(t) \eee^{\mi \psi_{1,N_1}(t)} &= \frac{1}{N_1} \sum_{k = 1}^{N_1} \eee^{\mi \theta_{1,k}(t)}, \label{eq:op1} \\ 
r_{2,N_2}(t) \eee^{\mi \psi_{2,N_2}(t)} &= \frac{1}{N_2} \sum_{l = 1}^{N_2} \eee^{\mi \theta_{2,l}(t)}. \label{eq:op2}
\end{align}
We call $r_{1,N_1}(t) \in [0,1]$ and $r_{2,N_2}(t) \in [0,1]$ the \emph{synchronization levels of community 1 and community 2}, respectively. Furthermore $\psi_{1,N_1}(t) \in \S$ and $\psi_{2,N_2}(t) \in \S$ represent the \emph{average phases of community 1 and 2}. Using these order parameters we can rewrite equations (\ref{kura1}) and (\ref{kura2}). Multiplying \eqref{eq:op1} and \eqref{eq:op2} with $\eee^{-\mi \theta_{1,i}(t)}$ and $\eee^{-\mi \theta_{2,i}(t)}$ respectively, taking the imaginary part of the resulting equations and plugging these into \eqref{kura1} and \eqref{kura2} gives
\begin{align}
&\ddd \theta_{1,i}(t)	= \omega_{1,i}\ddd t + \frac{K_1 N_1}{N_1 + N_2} r_{1,N_1}(t) \sin( \psi_{1,N_1}(t) - \theta_{1,i}(t) ) \ddd t \nonumber\\
	&+ \frac{L_1 N_2}{N_1 + N_2} r_{2,N_2}(t) \sin( \psi_{2,N_2}(t) - \theta_{1,i}(t) ) \ddd t + \sqrt{D} \ddd W_{1,i}(t), \label{km1}
\end{align}
and
\begin{align}
&\ddd \theta_{2,j}(t)	= \omega_{2,j}\ddd t + \frac{K_2 N_2}{N_1 + N_2} r_{2,N_2}(t) \sin( \psi_{2,N_2}(t) - \theta_{2,j}(t) ) \ddd t \nonumber\\
	&+ \frac{L_2 N_1}{N_1 + N_2} r_{1,N_1}(t) \sin( \psi_{1,N_1}(t) - \theta_{2,j}(t) ) + \sqrt{D} \ddd W_{2,j}(t). \label{km2}
\end{align}

Note that the model is \emph{rotationally invariant}, this means that if $\theta_{1,i}(t)$ is a solution of \eqref{km1}, then $\theta_{1,i}(t) + \theta_0$ is also a solution of \eqref{km1}, for any constant $\theta_0 \in \S$ and $i = 1,\ldots, N_1$. Similarly if $\theta_{2,j}(t)$ is a solution of \eqref{km2}, then $\theta_{2,j}(t) + \theta_0$ is also a solution of \eqref{km2}, for any constant $\theta_0 \in \S$ and $j = 1,\ldots, N_2$.
\\

Furthermore we can assume without loss of generality that $\mu_1$ and $\mu_2$ have mean zero since we can map the model for each community to a model which rotates with speed $\int_\R \omega \mu_1(\ddd \omega)$ and $\int_\R \omega \mu_2(\ddd \omega)$ respectively, as in \cite{G10}.

By defining empirical measures for each community and taking the limit as $N$ tends to infinity we can derive the McKean-Vlasov equation for the system. We set $N_1 = \a_1 N$ and $N_2 = \a_2 N$, with $\a_1 + \a_2 = 1$ and define the \emph{empirical measure} for each community:

\begin{equation}
\nu_{N_1,t}	= \frac{1}{N_1} \sum_{i = 1}^{N_1} \delta_{(\theta_{1,i}(t), \omega_{1,i})},\;\; \text{and }\;\;\nu_{N_2,t}	= \frac{1}{N_2} \sum_{j = 1}^{N_2} \delta_{(\theta_{2,j}(t), \omega_{2,j})}.\nonumber
\end{equation}
\begin{theorem}[McKean-Vlasov limit]
\label{thm:mckean}
In the limit $N \to \infty$, the empirical measure $\nu_{N_1,t}$ converges to $p_1$, and the empirical measure $\nu_{N_2,t}$ converges to $p_2$, where $p_k(t;\theta, \omega)$ is evolves according to

\begin{equation}
\label{eq:McKean1}
\frac{\p p_k(t;\theta, \omega)}{\p t} = \frac{D}{2} \frac{\p^2 p_k(t;\theta, \omega)}{\p \theta^2} - \frac{\p}{\p \theta} [ v_k(t;\theta, \omega) p_k(t;\theta, \omega)],
\end{equation}
with 
\begin{align}
v_k(t;\theta, \omega) = \omega &+ \a_k K_k r_k(t) \sin( \psi_k(t) - \theta )\\
 &+ \a_{k'} L_k r_{k'}(t) \sin(\psi_{k'}(t) - \theta),\nonumber
\end{align}
where $k\in\{1, 2\}$ and $k'$ is the complement of $k$.
Here $r_k(t)$ and $\psi_k(t)$ are defined by

\begin{equation}
r_k \eee^{\mi \psi_k} = \int_\R \int_\S \eee^{\mi \theta } p_k(\theta, \omega) \ddd \theta \ddd \mu_k(\omega), \label{nop1}
\end{equation}
for $k\in \{1, 2\}$.
\end{theorem}
The proof of Theorem \ref{thm:mckean} is analogous to the proof in the case of the one-community noisy Kuramoto model. We refer to \cite[Chapter 10]{H00} for details.

In the next proposition we derive the stationary densities for the dynamics governed by \eqref{eq:McKean1}.
\begin{proposition}[Stationary solutions]\label{stat}
Suppose that $r_1 = r_2 = 0$ or $r_1, r_2 > 0$. The stationary density, $p_k(\theta, \omega)$ of community $k$, solves the equation

\begin{equation}
\frac{D}{2} \frac{\p^2 p_k(\theta, \omega)}{\p \theta^2} - \frac{\p}{\p \theta}[ v_k(\theta, \omega) p_k(\theta, \omega) ] = 0, \label{se1}
\end{equation}
which has a solution 
\begin{equation}
p_k(\theta,\omega) = \frac{A_k(\theta,\omega)}{\int_\S A_k(\phi, \omega) d \phi}, \label{p1}
\end{equation}
where 
\begin{align}
A_k(\theta, \omega) = B_k(\theta, \omega) &\Big[ \eee^{\frac{4 \pi \omega}{D}} \int_\S \frac{\ddd \phi}{B_k(\phi, \omega)} \\
&+ (1 - \eee^{\frac{4 \pi \omega}{D}}) \int_0^\theta \frac{\ddd \phi}{B_k(\phi,\omega)} \Big],\nonumber
\end{align}
with 
\begin{align}
B_k(\theta, \omega) = \exp \Big[ \frac{2 \omega \theta}{D} &+ \frac{2 \a_k K_k r_k \cos(\psi_k - \theta)}{D} \\
&+ \frac{2 \a_{k'} L_k r_{k'} \cos(\psi_{k'} - \theta)}{D}  \Big],\nonumber
\end{align}
where $k\in\{1, 2\}$ and $k'$ is the complement of $k$.
In addition, the following self-consistency equations must be satisfied
\begin{align}
r_k &= \int_\R  \int_\S \cos(\psi_k - \theta) p_k(\theta, \omega) \ddd \theta \ddd \mu_k( \omega), \label{eq:gsc1} \\
0 &= \int_\R  \int_\S \sin(\psi_k - \theta) p_k(\theta, \omega) \ddd \theta \ddd \mu_k(\omega), \label{eq:gsc3}
\end{align}
for $k\in \{1, 2\}$.
\end{proposition}
The proof of Proposition \ref{stat} is analogous to the calculation in \cite[Solution to Exercise X.33]{H00}.

\begin{remark}
Note that in the case of the one-community Kuramoto model one has $p(-\theta, - \omega) = p(\theta, \omega)$. This is in general not true for the two-community Kuramoto model.
\end{remark}

The self-consistency equations from the previous proposition cannot be solved explicitly. We can however simplify equation \eqref{eq:gsc1} in a simplified setting. To this end we assume that $\a_1 = \a_2$, $D = 1$ and $\mu_1 = \mu_2 = \delta_0$ and omit $\omega \in \S$ in the notation.
Under these assumptions the stationary densities, \eqref{p1}, simplify to

\begin{align}
p_k(\theta) &= \frac{B_k(\theta)}{\int_\S B_k(\phi) \ddd \phi}\\
 &=  \frac{\exp\left[ L_k r_{k'} \cos(\psi_{k'} - \theta) + K_k r_k \cos(\psi_k - \theta)\right]}{\int_\S \exp\left[ L_k r_{k'} \cos(\psi_{k'} - \phi) + K_k r_k \cos(\psi_k - \phi) \right] \ddd \phi}.\nonumber
\end{align}

Furthermore the self-consistent equations, \eqref{eq:gsc1}, simplify to
\begin{align}
r_k &= \frac{\int_\S \cos(\psi_k - \theta) \eee^{\left[ L_k r_{k'} \cos(\psi_{k'} - \theta) + K_k r_k \cos(\psi_k - \theta)\right]} \ddd \theta }{\int_\S \eee^{\left[ L_k r_{k'} \cos(\psi_{k'} - \phi) + K_k r_k \cos(\psi_k - \phi) \right]} \ddd \phi}. \label{eq:osc1}
\end{align}
In addition, the self-consistent equations, \eqref{eq:gsc3} simplify to
\begin{align}
0 &= \frac{\int_\S \sin(\psi_k - \theta) \eee^{\left[ L_k r_{k'} \cos(\psi_{k'} - \theta) + K_k r_k \cos(\psi_k - \theta)\right]} \ddd \theta }{\int_\S \eee^{\left[ L_k r_{k'} \cos(\psi_{k'} - \phi) + K_k r_k \cos(\psi_k - \phi) \right]} \ddd \phi}. \label{eq:osc3}
\end{align}

Let us consider $K_1, K_2 \in \R$ positive and $L_1, L_2 \in \R \setminus \{0\}$. We will first rewrite equations \eqref{eq:osc1} and \eqref{eq:osc3} to a more convenient form. In order to do this we define the following functions
\begin{definition}[Special functions]
\begin{equation}
V(x) := \frac{\int_{\S}\cos \theta\eee^{x \cos \theta}\ddd\theta}{\int_{\S}\eee^{x \cos \theta}\ddd\theta}\quad \text{and }\quad W(x) := \frac{2V(x)}{x}.
\end{equation}
\end{definition}
Properties of $V(x)$ and $W(x)$ are given in Lemma \ref{lV} and Lemma \ref{thm:W} respectively.

\begin{proposition}[Self-consistency equations]\ \\
\label{prop:selfcons}
The self-consistency equations \eqref{eq:osc1} and \eqref{eq:osc3} can be rewritten as
\begin{align}
r_k =& \frac{K_k r_k + L_k r_{k'} \cos \psi}{2} \nonumber\\
&\times W \left( \sqrt{ K_k^2 r_k^2 + L_k^2 r_{k'}^2 + 2 K_k L_k r_k r_{k'} \cos \psi} \right), \label{eq:sc1}
\end{align}
and 
\begin{equation}
0 = L_k r_{k'} \sin \psi W \left( \sqrt{ K_k^2 r_1^2 + L_k^2 r_{k'}^2 + 2 K_k L_k r_k r_{k'} \cos \psi} \right),\label{eq:sc3}
\end{equation}
with $x \in (0,\infty)$.
\end{proposition}
The proof of Proposition \ref{prop:selfcons} is given in Appendix \ref{app:selfcons}. The number of possible solutions to the system of equations given in Proposition \ref{prop:selfcons} can be reduced significantly using the following theorem.

\begin{theorem}[Reduction of possible solutions]\ \\
\label{0r}
\vspace{-0.5cm}
\begin{enumerate}
\item Solutions of the form $r_1 = 0$ and $r_2 > 0$, or vice-versa, do not exist.
\item If $r_1, r_2 > 0$, then $\psi \in \{0, \pi \}$.
\end{enumerate}
\end{theorem}

\begin{proof} (Claim 1) Suppose that $r_1 = 0$ and $r_2 > 0 $, then equations (\ref{eq:sc1})-(\ref{eq:sc3}) reduce to
\begin{align}
0	&= \frac{L_1 r_2 \cos \psi}{2} W(L_1 r_2),\label{eq:nr1}\\ 
r_2	&= \frac{K_2 r_2}{2} W(K_2 r_2),\\
0	&= L_1 r_2 \sin \psi W(L_1 r_2). \label{eq:nr2}
\end{align}
Note that equation \eqref{eq:sc3} with $k=2$ is always satisfied. Also, by assumption, $L_1 \neq 0$. In addition, by Lemma \ref{thm:W}, we have $W > 0$. In order to satisfy equations \eqref{eq:nr1} and \eqref{eq:nr2} simultaneously, it must hold that there exists some $\psi \in [0, 2\pi)$ such that $\sin \psi = 0$ and $\cos \psi = 0$, which is not possible. 

(Claim 2) We have that equation \eqref{eq:sc3} reduces to $\sin \psi = 0$, and therefore $\psi \in \{0, \pi\}$, since $L_1, L_2 \neq 0$ and $W > 0$.
\end{proof}

\begin{remark}[Reduced self-consistency equations]
If $\psi \in \{0, \pi\}$, then the self-consistency equations of Proposition \ref{prop:selfcons} reduce to 
\begin{align}
r_k	&= \frac{K_k r_k + L_k r_{k'} \cos \psi }{2} W( K_k r_k + L_k r_{k'} \cos \psi) \nonumber\\
&= V(K_k r_k + L_k r_{k'} \cos \psi), \label{eq:scp1}
\end{align}
The other two self-consistency equations , \eqref{eq:sc3}, are always satisfied when $\psi \in \{0, \pi\}$. Note that $\cos \psi \in \{-1,1\}$ for $\psi \in \{0, \pi\}$.
\end{remark}

It might seem surprising that the phase difference between the average phases of the two communities is restricted to take the values $0$ and $\pi$. Intuitively the phase difference changes the effective interaction strength between the two communities. From energetic considerations one might argue that the system would evolve to either maximize or minimize this interaction strength. This maximization or minimization is achieved precisely when the phase difference is $0$ or $\pi$.  

\section{Overview of the parameter space}
\label{sec:overview}
From this point onward we set $\psi=0$ and note that the analysis is repeatable for the case $\psi=\pi$ with the modification $L_{1}\rightarrow -L_{1}$ and $L_{2}\rightarrow -L_{2}$. The equations in Proposition \ref{prop:selfcons} may have multiple solutions. The number of solutions varies with the model parameters, so that it is possible to delineate the regions in the parameter space with the same number of solutions. With this goal in mind, we split the parameter space into regions based on a geometric interpretation of the self-consistency equations.

\subsection{Level curves}
To visualize the self-consistency equations we define the following functions 
\begin{definition}[Self-consistency surfaces]
\label{def:hsurfaces}
Let $h_1^{K_1, L_1}$ and $h_2^{K_2, L_2}:[0,1]^2 \to \R$ be given by
\begin{align}
h_1^{K_1, L_1}(r_1, r_2) &:= V(K_1 r_1 + L_1  r_2) - r_1,\\
h_2^{K_2, L_2}(r_1, r_2) &:= V(K_2 r_2 + L_2  r_1) - r_2. 
\end{align}
Additionally, denote by $h_1^{K_1, L_1, +}$ the non-negative part of $h^{K_1,L_1}_1$ and by $h_2^{K_2, L_2, +}$ the non-negative part of $h^{K_2,L_2}_2$.
\end{definition}

When the context is clear, we omit the $K_1, L_1$ and $K_2, L_2$ in the notation and write $h_1 = h_1^{K_1, L_1}$ and $h_2 = h_2^{K_2, L_2}$.

In order for a pair of synchronization levels $(r_{1}, r_{2})$ to satisfy the self-consistency equations \eqref{eq:scp1} (for a set of parameter values $K_{1}$, $K_{2}$, $L_{1}$ and $L_{2}$) it must be a point such that these surfaces intersect one another as well as the zero plane. We can thus restrict ourselves to the curves defined by the intersection of the self-consistency surfaces with the zero plane defined by
\begin{definition}[Self-consistency intersection curves]
\label{def:sccurves}
\begin{align}
\Gamma_1^{K_1, L_1} &:= \left\{ (r_1, r_2) \in [0,1]^2 : h_1^{K_1, L_1}(r_1, r_2) = 0 \right\},\\
\Gamma_2^{K_2, L_2} &:= \left\{ (r_1, r_2) \in [0,1]^2 : h_2^{K_2, L_2}(r_1, r_2) = 0 \right\}.
\end{align}
\end{definition}
Again, when the context is clear, we write $\Gamma_1 = \Gamma_1^{K_1, L_1}$, $\Gamma_2 = \Gamma_2^{K_2, L_2}$. Note that $\Gamma_1^{K_1, L_1} \cap \Gamma_2^{K_2, L_2}$ is the set of solutions of the self-consistency equations \eqref{eq:scp1}. Hence the solutions to the self-consistency equation are precisely the intersections of the self-consistency curves $\Gamma_1$, $\Gamma_2$.
Clearly $(0,0) \in \Gamma_1^{K_1, L_1}, \Gamma_2^{K_2, L_2}$ for all $K_1,K_2 \in \R$ and $L_1, L_2 \in \R\setminus \{0\}$,  since $(0, 0)$ always solves the self-consistency equations. A visualization of the curves defined in Definition \ref{def:sccurves} in various regions is shown in Figure \ref{numregov}.

\begin{figure*}
 \centering
 	% row 1
 	\begin{subfigure}{0.3\textwidth}
        \centering
        \includegraphics[scale = 0.43]{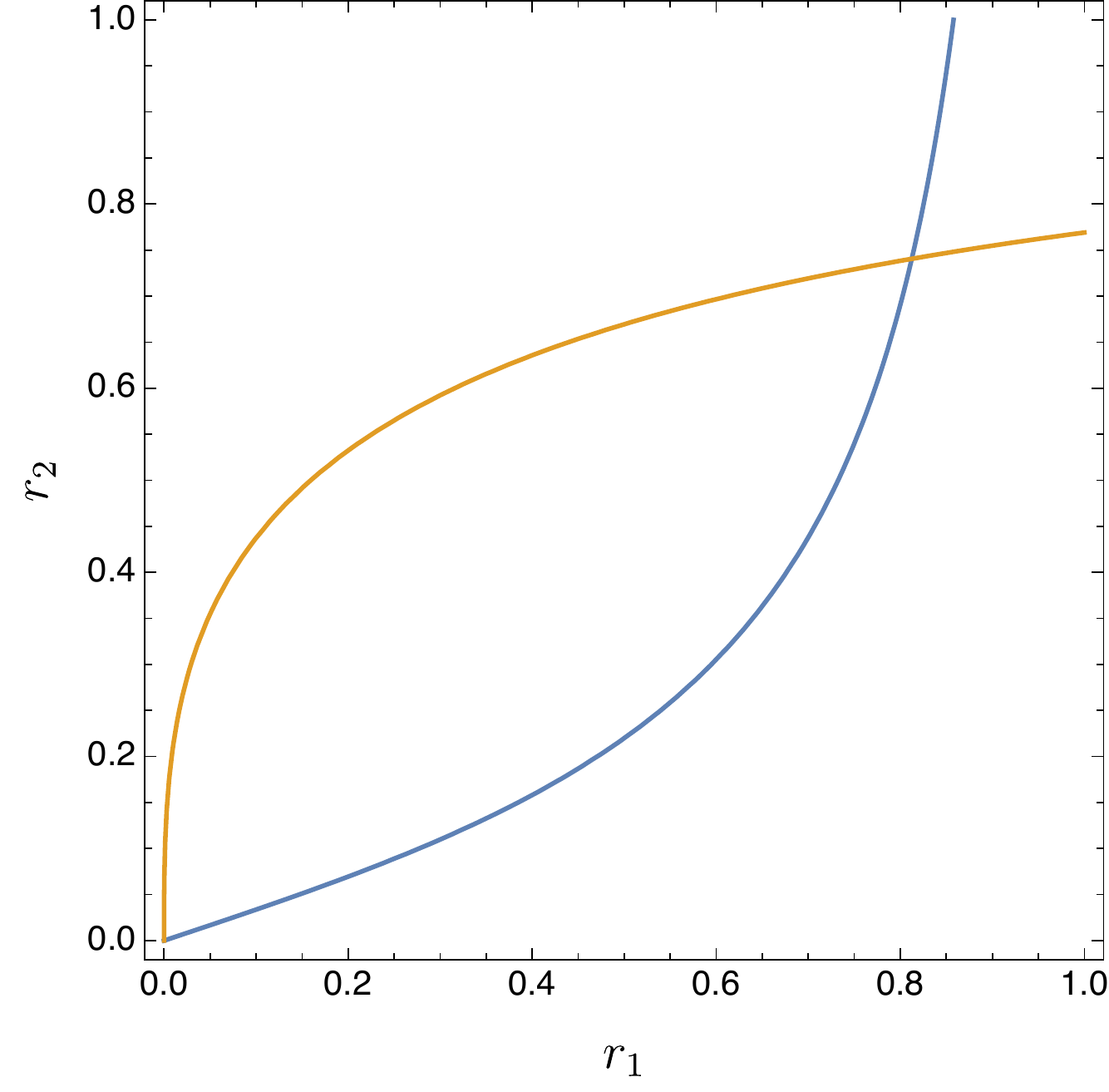}
        \caption{{Region 2: $K_1 = 1, K_2 = 2$ and} ${L_1 = 3}, L_2 = 1.$}
    \end{subfigure}\hfill
    \begin{subfigure}{0.3\textwidth}
        \centering
        \includegraphics[scale = 0.43]{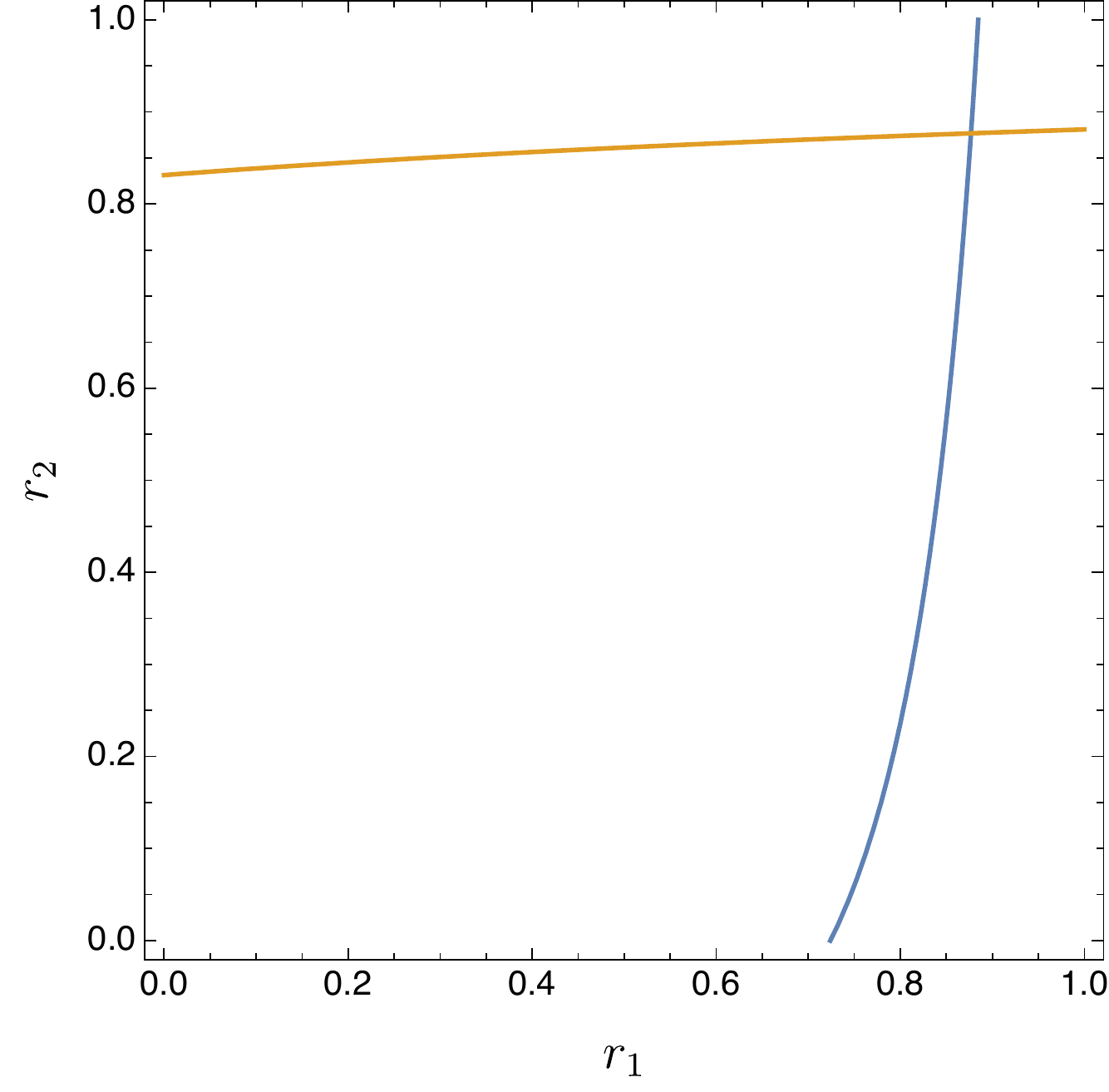}
        \caption{Region 3: $K_1 = 3, K_2 = 4$ and ${L_1 = 2}, L_2 = 1.$}
    \end{subfigure}\hfill
    \begin{subfigure}{0.3\textwidth}
        \centering
        \includegraphics[scale = 0.43]{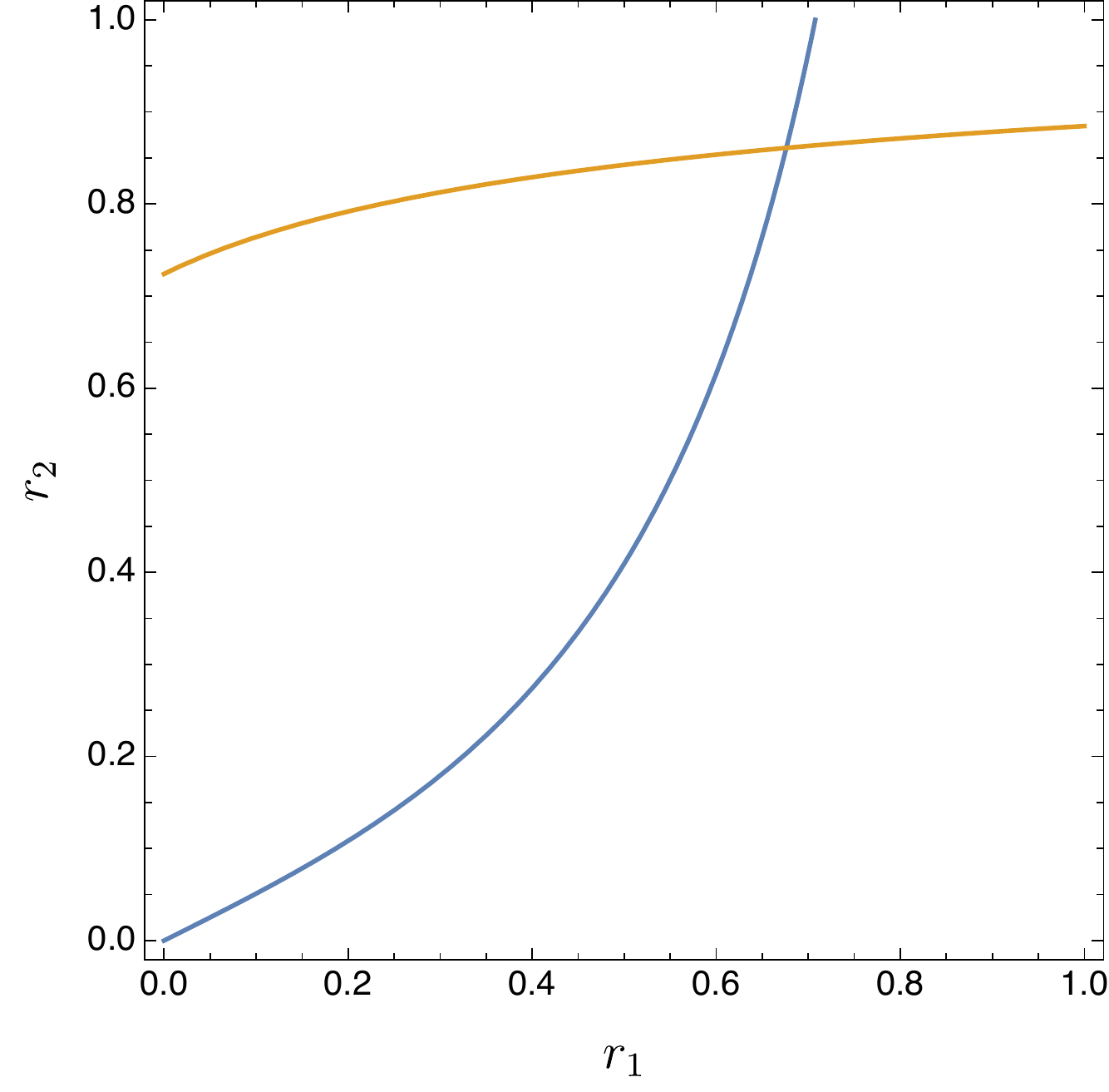}
        \caption{Region 4: $K_1 = 1.5, K_2 = 3$ and ${L_1 = 1}, {L_2 = 2.}$}
    \end{subfigure}\vspace{2 mm}
    
    % row 2
    \begin{subfigure}{0.3\textwidth}
        \centering
        \includegraphics[scale = 0.43]{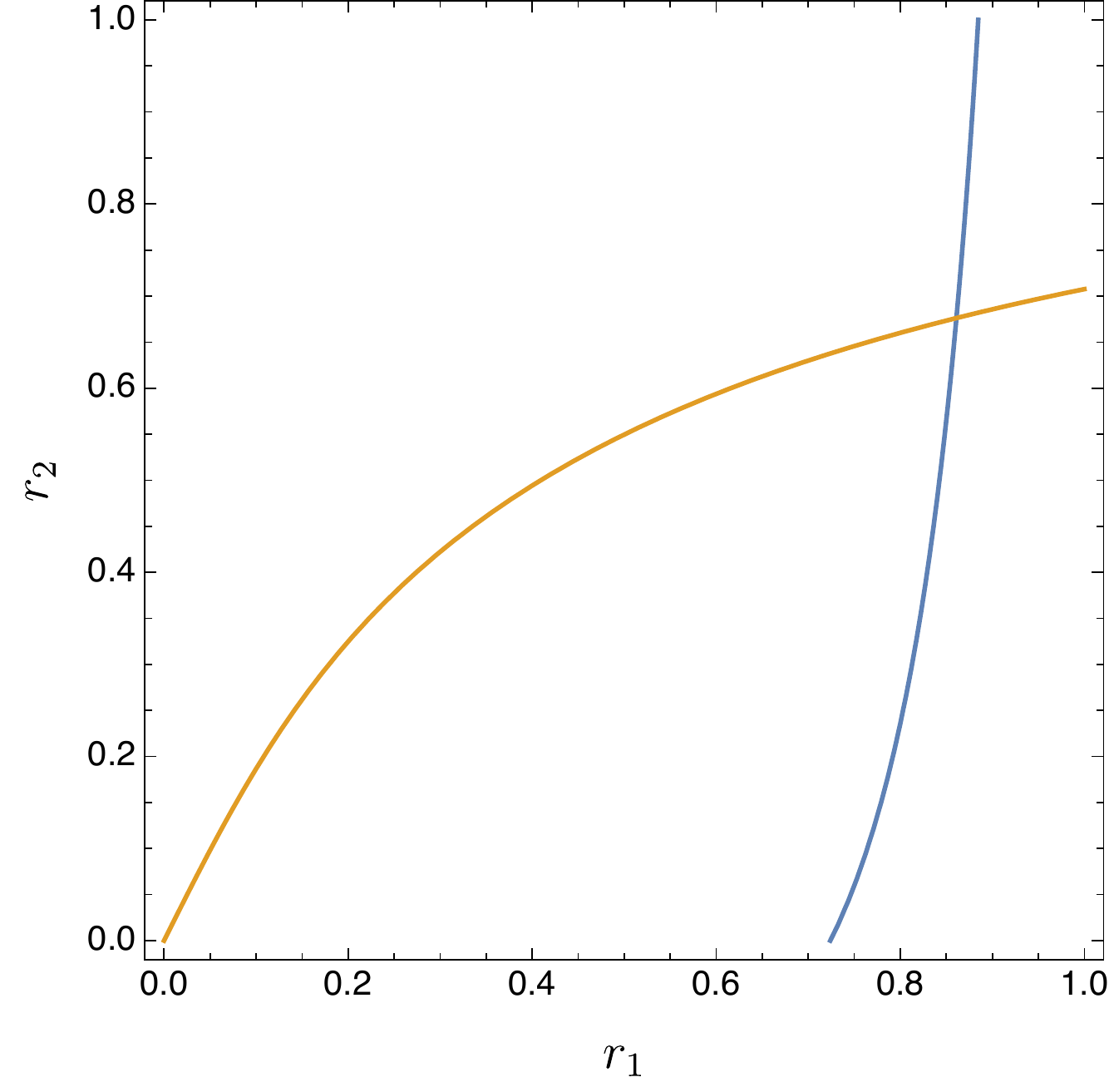}
        \caption{Region 5: $K_1 = 3, K_2 = 1.5$ and ${L_1 = 2, L_2 = 1.}$}
    \end{subfigure}\hfill
    \begin{subfigure}{0.3\textwidth}
        \centering
        \includegraphics[scale = 0.43]{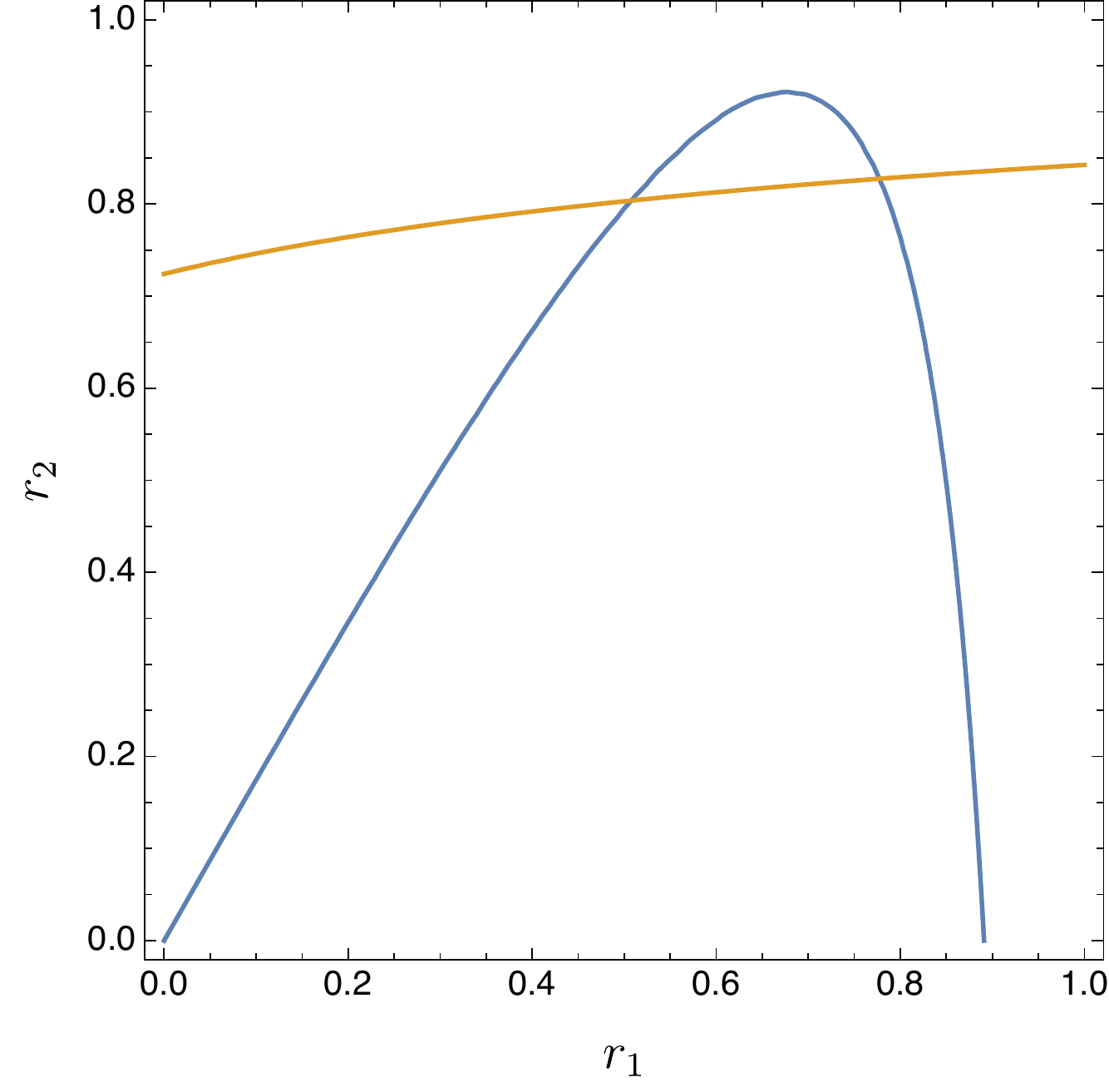}
        \caption{Region 6: $K_1 = 5.5, K_2 = 3$ and ${L_1 = -2, L_2 = 1.}$}
    \end{subfigure}\hfill
    \begin{subfigure}{0.3\textwidth}
        \centering
        \includegraphics[scale = 0.43]{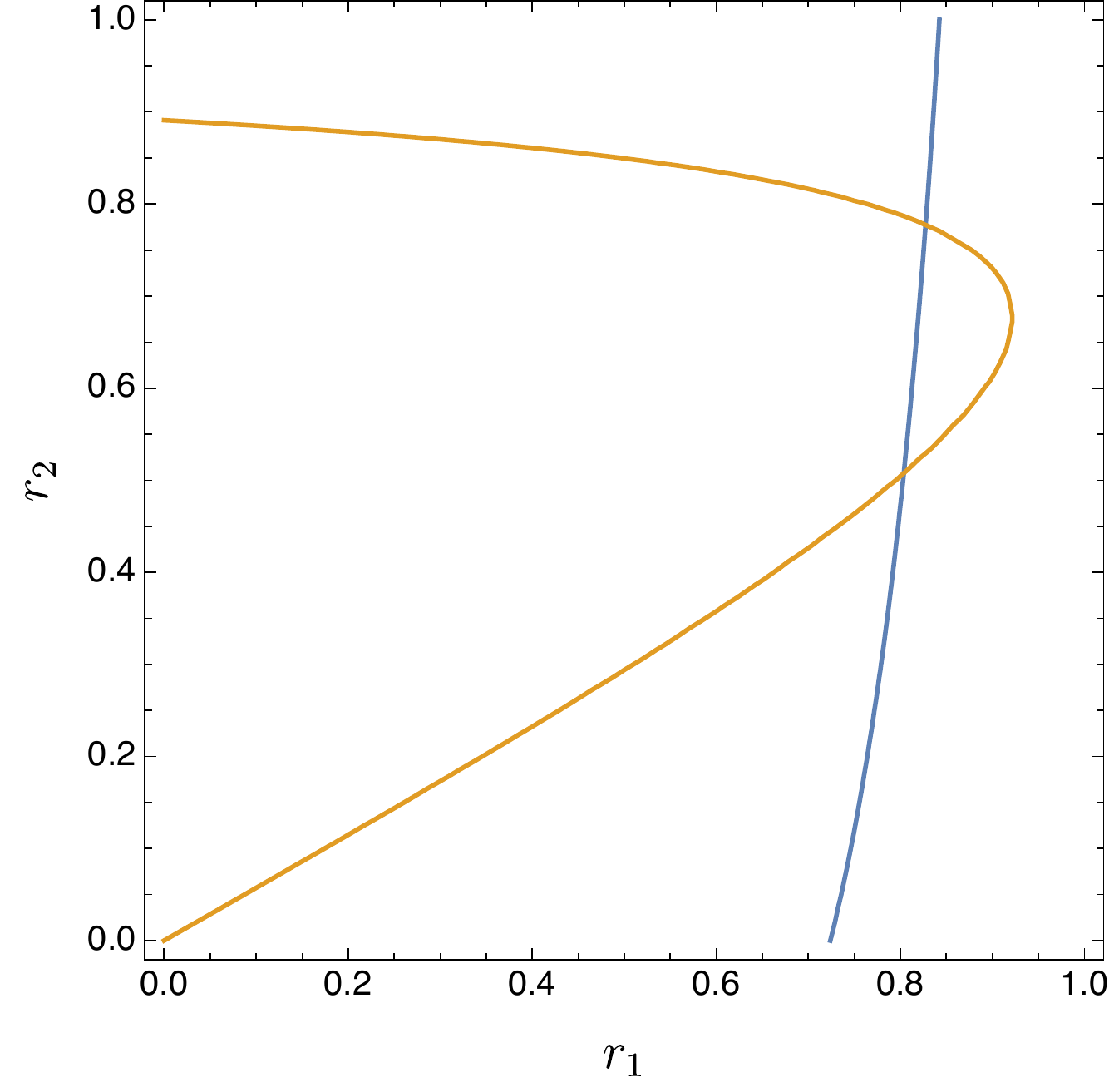}
        \caption{Region 7: $K_1 = 3, K_2 = 5.5$ and ${L_1 = 1, L_2 = -2.}$}
    \end{subfigure}\vspace{2 mm}
    
    % row 3
    \begin{subfigure}{0.3\textwidth}
        \centering
        \includegraphics[scale = 0.43]{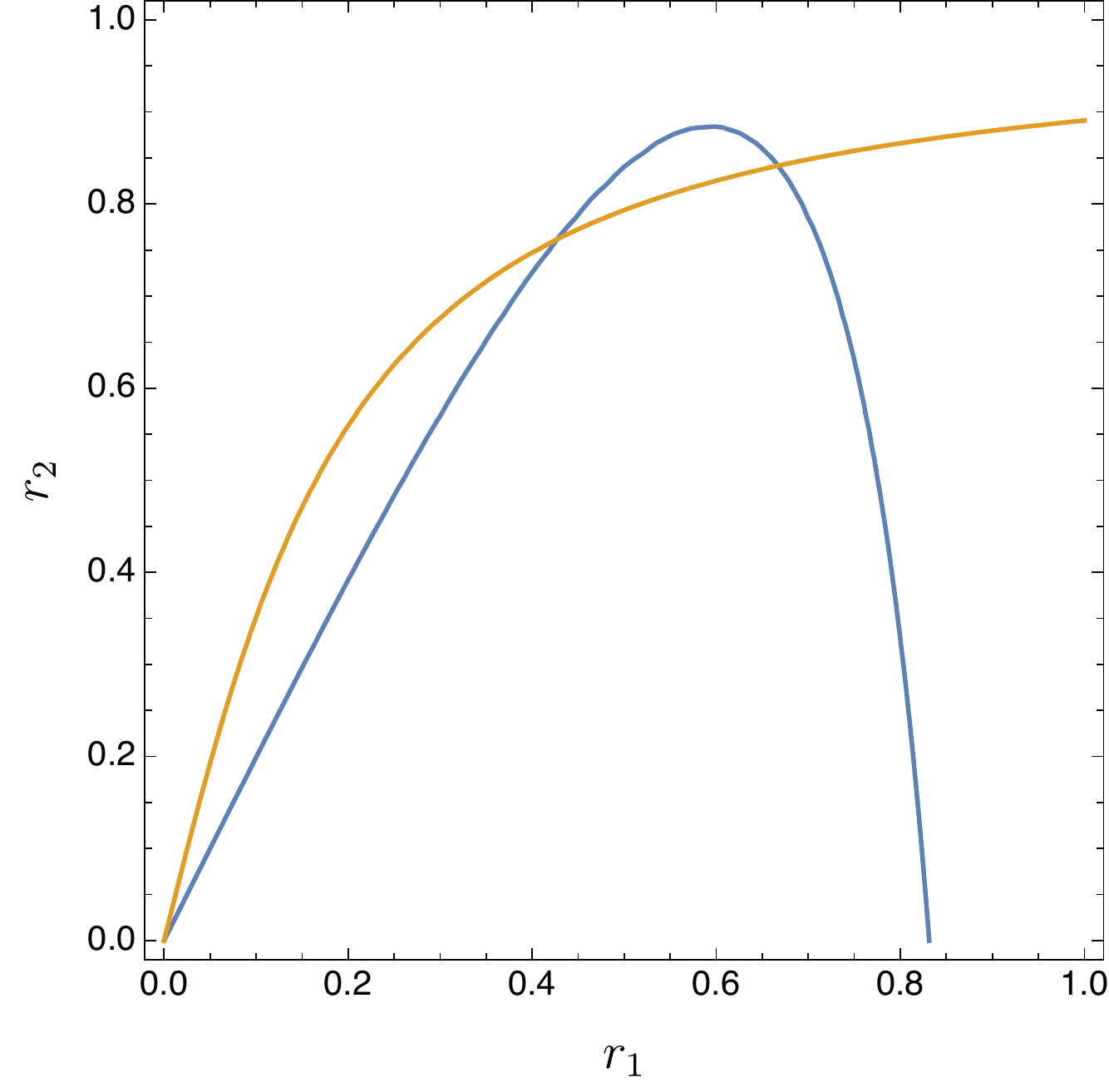}
        \caption{Region 8: $K_1 = 4, K_2 = 1$ and ${L_1 = -1, L_2 = 4}.$}
    \end{subfigure}\hfill
    \begin{subfigure}{0.3\textwidth}
        \centering
        \includegraphics[scale = 0.43]{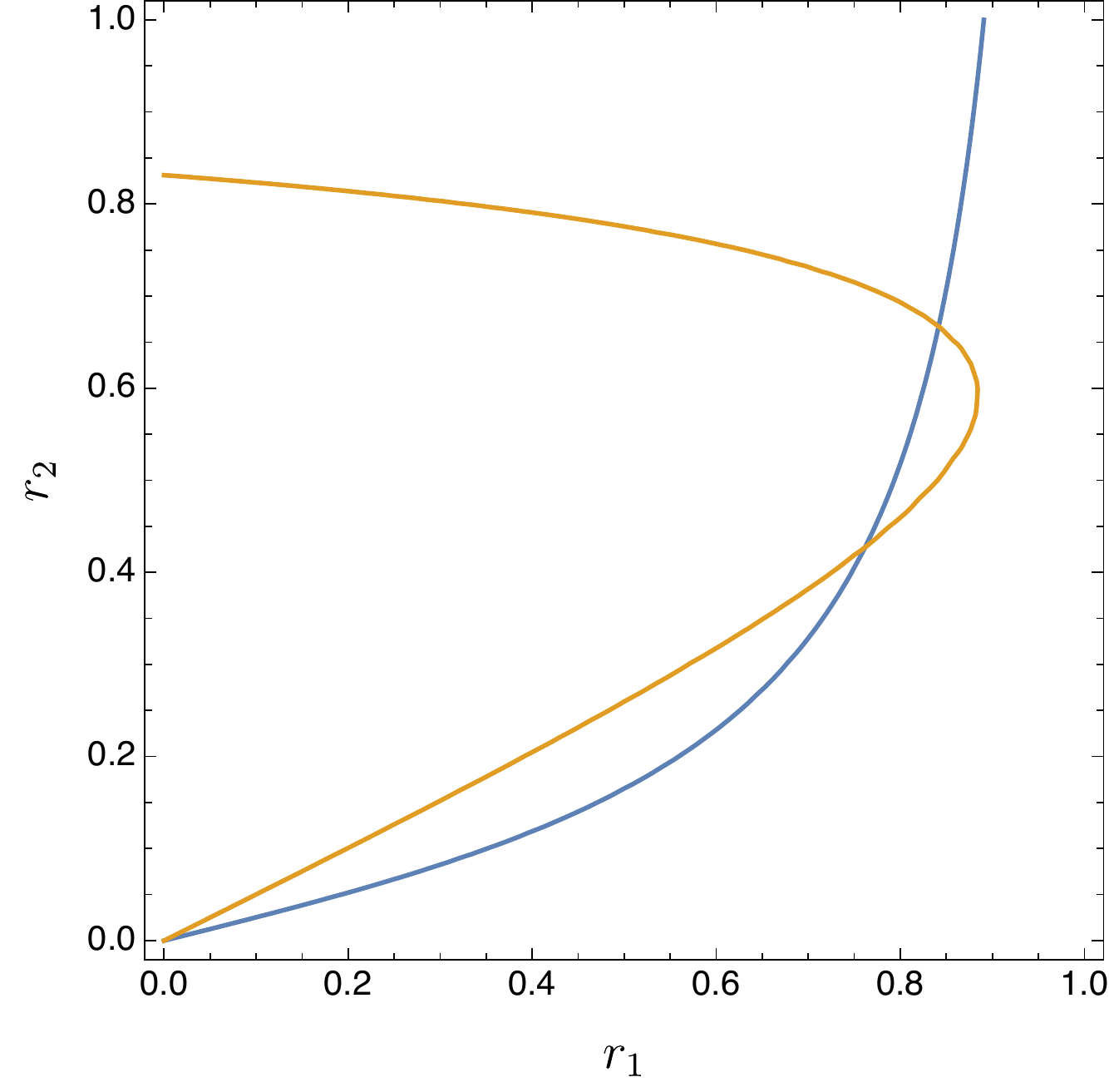}
        \caption{Region 9: $K_1 = 1, K_2 = 4$ and ${L_1 = 4, L_2 = -1}.$}
    \end{subfigure}\hfill
    \begin{subfigure}{0.3\textwidth}
        \centering
        \includegraphics[scale = 0.43]{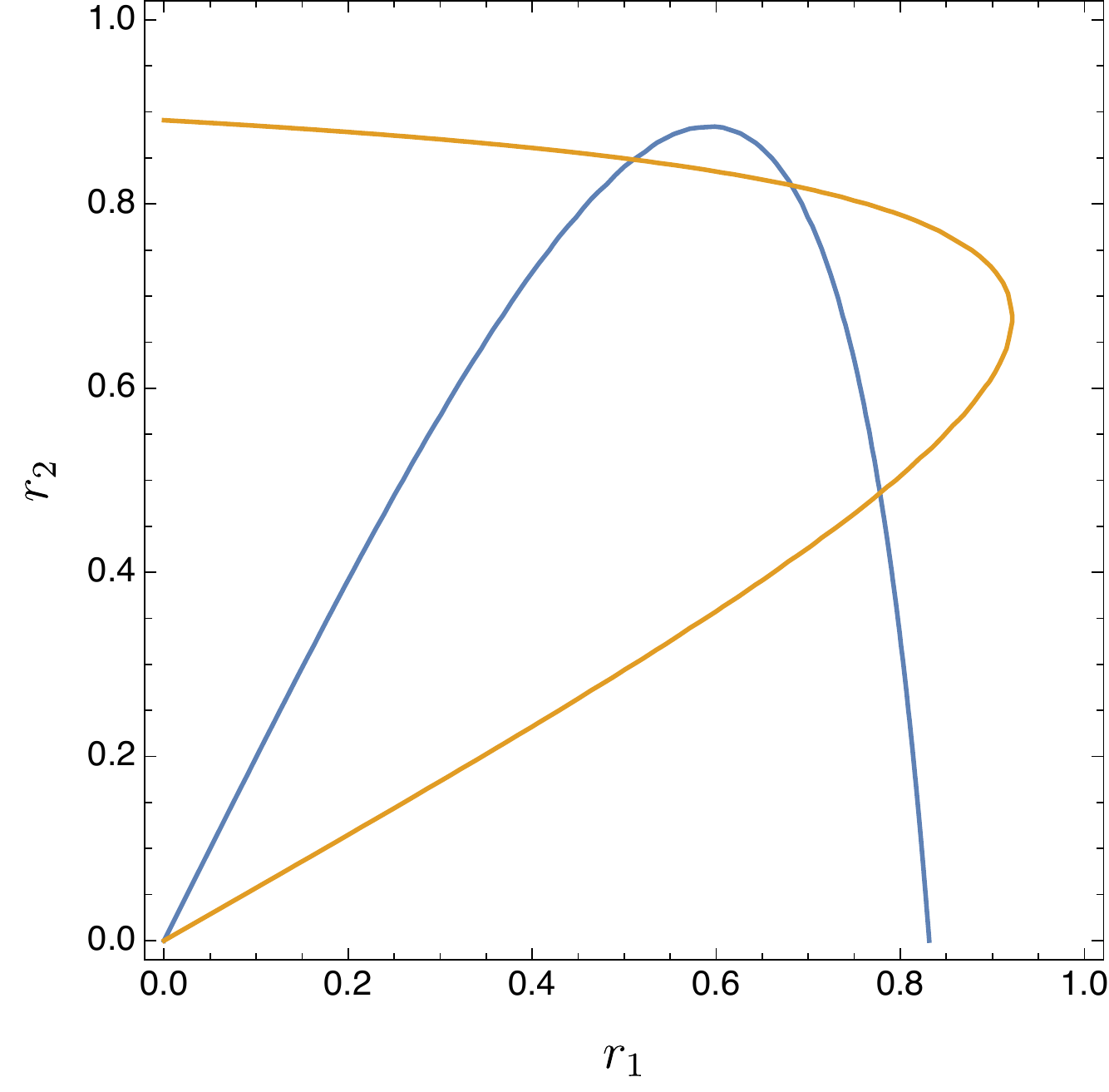}
        \caption{Region 10: $K_1 = 4, K_2 = 5.5$ and ${L_1 = -1, L_2 = -2.}$}
    \end{subfigure}
    
    \caption{For each of the regions 2-10 of Table \ref{fig:overview} a numerical example where the maximum number of intersections between $\Gamma_1$ and $\Gamma_2$ is shown. Note that region 1 is not displayed in this figure because in this case either $\Gamma_1$ or $\Gamma_2$ is trivial.}\label{numregov}
\end{figure*}

The self-consistency surfaces have different properties in different domains of the parameter space. This allows us to partition the parameter space into regions that are easier to analyze. The following theorem identifies a region of the parameter space in which only the unsynchronized solution is possible.

\begin{theorem}
\label{thm:hneg}
$h_1^{K_1, L_1}(r_1, r_2) < 0$ for all $(r_1,r_2) \in [0,1]^2\setminus \{(0,0)\}$ if and only if $K_1 \leq 2$ and $L_1 < 0$.
\end{theorem}
The proof of Theorem \ref{thm:hneg} is given in Appendix \ref{app:proofhneg}. 

Our next result presents properties of the derivatives of the self-consistency surfaces. In order to simplify the notation, we abbreviate the derivatives of $V(\cdot)$ evaluated at various points as 
\begin{align}
\label{eq:Cderivatives1}
C_{1,1} &= V'(K_1 r_1 + L_1 r_2),\quad C_{1,2} = V''(K_1 r_1 + L_1 r_2),\\
C_{2,1} &= V'(K_2 r_2 + L_2 r_1),\quad C_{2,2} = V''(K_2 r_2 + L_2 r_1).
\label{eq:Cderivatives2}
\end{align}

\begin{lemma}[Derivatives of $h_1$]
\label{lem:derivativesofh1}
Take $r_1, r_2 \neq 0$, then
\begin{enumerate}
\item $\frac{\p^2 h_1}{\p r_1^2} < 0$ and $\frac{\p^2 h_1}{\p r_2^2} < 0$,
\item $\frac{\p h_1}{\p r_2} > 0$ if and only if $L_1  > 0$,
\item $\frac{\p h_1}{\p r_1}(0,0) > 0$ if and only if $K_1 > 2$.
\end{enumerate}
\end{lemma}

\begin{proof}
Taking partial derivatives of $h_{1}$ gives:
\begin{align}
\frac{\p h_1}{\p r_1}(r_1, r_2)&= K_1 C_{1,1} - 1,\quad &\frac{\p^2 h_1}{\p r_1^2}(r_1, r_2) = K_1^2 C_{1, 2},\\
\frac{\p h_1}{\p r_2}(r_1, r_2) &= L_1  C_{1, 1},\quad &\frac{\p^2 h_1}{\p r_2^2}(r_1, r_2) = L_1^2 C_{1,2}.
\end{align}
Since $V'' < 0$ (Property 5 of Lemma \ref{lV}) we clearly have $\frac{\p^2 h_1}{\p r_1^2} < 0$ and $\frac{\p^2 h_1}{\p r_2^2} < 0$. In addition, $\frac{\p h_1}{\p r_2} > 0$ if and only if $L_1 > 0$ because $V' > 0$ (Property 4 of Lemma \ref{lV}). Furthermore
\begin{equation}
\frac{\p h_1}{\p r_1}(0,0) = \frac{K_1}{2} - 1,
\end{equation}
which is positive if and only if $K_1 > 2$.
\end{proof}

\begin{remark}
Note that $\frac{\p^2 h_1}{\p r_1^2}(r_1, r_2)  < 0$ and $\frac{\p^2 h_1}{\p r_2^2}(r_1, r_2) < 0$ for all $K_1 \in \R$, $L_1 \in \R\setminus \{0 \}$ and $(r_1, r_2) \neq (0,0)$, so that $h_1$ is strictly concave.
\end{remark}

Theorem \ref{thm:hneg} and Lemma \ref{lem:derivativesofh1} show that there are three domains in which the behavior of $h_1^+$ differs. By Theorem \ref{thm:hneg}, we have that $h_1^+ = 0$ if $K_1 < 2$ and $L_1 < 0$, and that $h_1^+ \neq 0$ in the complement of this region. Using Lemma \ref{lem:derivativesofh1} we partition the complement into three domains.

\begin{definition}[Curve domains]\ \\
\label{def:curvedomains}
\vspace{-0.5cm}
\begin{enumerate}
\item $K_1 \leq 2$ and $L_1  > 0$,
\item $K_1 > 2$ and $L_1 > 0$,
\item $K_1 > 2$ and $L_1 < 0$.
\end{enumerate}
\end{definition}
Each domain corresponds to a curve for $\Gamma_1^{K_1, L_1}$ with different characteristics. We refer to the three possible curves as the \emph{fundamental curves}. The same splitting of the parameter space can be done when considering $h_2^{K_2, L_2}(r_1, r_2)$. Analyzing the number of solutions for $(r_{1},r_{2})$ that are possible in a region of the parameter space resulting from the combination of two domains (one for $K_{1}$ and $L_{1}$ and one for $K_{2}$ and $L_{2}$) is equivalent to analyzing the number of intersections possible between the two fundamental curves corresponding to the domains.

Based on Theorem \ref{thm:hneg} and Lemma \ref{lem:derivativesofh1} we can derive the defining properties of the fundamental curves. By implicit differentiation we determine the derivatives along the curve $\Gamma_1^{K_1, L_1}$. Denote by $\frac{\p \Gamma_1}{\p r_1}$ the derivative $\dot{r_2}(r_1)$ along the curve $\Gamma_1$.

\begin{lemma}[Derivatives of $\Gamma_1$ and $\Gamma_2$]\label{lem:pderivatives} The derivatives on $\Gamma_1$ and $\Gamma_2$ with respect to $r_1$ are
\begin{enumerate}
\item $\frac{\p \Gamma_1}{\p r_1} = \frac{1 - K_1 C_{1,1}}{L_1 C_{1,1}},$
\item $\frac{\p^2 \Gamma_1}{\p r_1^2} = - \frac{C_{1,2}}{L_1 C_{1,1}^3},$
\item $\frac{\p \Gamma_2}{\p r_1} = \frac{L_2 C_{2,1}}{1 - K_2 C_{2,1}},$
\item $\frac{\p^2 \Gamma_2}{\p r_1^2} = \frac{L_2^2 C_{2,2}}{(1 - K_2 C_{2,1})^3},$
\end{enumerate}
where $C_{1,1}, C_{1,2}, C_{2,1}$ and $C_{2,2}$ are given in \eqref{eq:Cderivatives1} and \eqref{eq:Cderivatives2}.
\end{lemma}
The proof of Lemma \ref{lem:pderivatives} is given in Appendix \ref{app:pderivatives}.

\begin{remark}
For $K_1 = 2$ and $L_1 > 0$ the denominators of $\frac{\p \Gamma_2}{\p r_1}(0,0)$ and $\frac{\p^2 \Gamma_2}{\p r_1^2}(0,0)$ are zero. In this case we set both equal to $\infty$.
\end{remark}

\begin{lemma}[Continuity of $\Gamma_1$]\label{lem:ccurve}
Assume that $K_1$ and $L_1$ lie outside the region $K_1 < 2$ and $L_1 < 0$. Then $\Gamma_1$ is continuous for every $(r_1, r_2) \in \Gamma_1 \setminus \{(0,0)\}$.
\end{lemma}

\begin{proof}
The claim follows from the fact that $\p \Gamma_1/\p r_1$ is defined for all $(r_1, r_2) \in [0,1]^2\setminus\{ (0,0)\}$. Hence the latter is also true for all $(r_1, r_2) \in \Gamma_1 \setminus \{ (0,0)\}$.
\end{proof}

There is a slight abuse of notation in Lemma \ref{lem:ccurve}, because we treat $\Gamma_1$ as a set and as a curve. The definition of $\Gamma_1$ as a set is given in Definition \ref{def:sccurves}. The curve $\Gamma_1$ is the implicit function $r_2(r_1)$ which solves $V(K_1 r_1 + L_1 r_2(r_1) ) - r_1 = 0$.  We remove the point $(0,0)$ from the set $\Gamma_1$ to make $r_2(r_1)$ a well-defined function.

\begin{remark}[Convexity/Concavity of $\Gamma_{1}$]
Assume that $K_1$ and $L_1$ lie outside the region $K_1 < 2$ and $L_1 < 0$. Then  $C_{1,1} > 0$ and $C_{2,1} < 0$ for all $(r_1, r_2) \in \Gamma_1 \setminus \{0,0\}$. Hence, by Lemma \ref{lem:pderivatives} it follows that 
\begin{equation}
\frac{\p^2 \Gamma_1}{\p r_1^2} < 0 \iff L_1 < 0.
\end{equation}
Thus, $\Gamma_1$ is strictly concave for $L_1 < 0$ and strictly convex for $L_1 > 0$.
\end{remark}

The following theorem identifies the defining characteristics of the fundamental curves. The first characteristic is whether the curve has a turning point or not, the second is whether the curve is connected with zero or not and the last is whether the curve is connected with a synchronization value of one. 

\begin{theorem}[Properties of $\Gamma_{1}$]\ \\
\label{thm:propfundcurves}
\vspace{-0.5cm}
\begin{enumerate}
\item $\frac{\p \Gamma_1}{\p r_1}$ changes sign if and only if $K_1 > 2$ and $L_1 < 0$.
\item In the domains $K_1 > 2,L_1 < 0$ and $K_1 \leq 2,L_1 > 0$
the curve $\Gamma_1$ is connected with the point $(r_1, r_2) = (0, 0)$.
\item In the domains $K_1 > 2,L_1 > 0$ and $K_1 \leq 2,L_1 > 0$
there exists precisely one $r \in (0,1)$ such that the curve $\Gamma_1$ is connected with the point $(r_1, r_2) = (r,1)$.
\end{enumerate}
\end{theorem}
The proof of Theorem \ref{thm:propfundcurves} is given in Appendix \ref{app:propfundcurves}.

Based on the preceding theorem we present the three fundamental curves.

\subsubsection{``Convex curve connected with zero"}
We restrict $\Gamma_1^{K_1, L_1}$ to the first domain of Definition \ref{def:curvedomains}, i.e., $K_1 \leq 2$, $L_1 > 0$. In this case the curve has the following properties.

\begin{enumerate}
\item By Property 1 of Theorem \ref{thm:propfundcurves} the curve has no turning point.
\item By Property 2 of Theorem \ref{thm:propfundcurves} the curve $\Gamma_1$ is connected with $(r_1, r_2) = (0,0)$,
\item By Property 3 of Theorem \ref{thm:propfundcurves} $\Gamma_1$ is connected with $(r_1, r_2) = (r,1)$, for some $r \in (0,1)$,
\end{enumerate}

Examples of $\Gamma_1$ in this domain are given in Figures \ref{fig:s13} and \ref{fig:s14} showing how the curve changes as the parameters $K_{1}$ and $L_{1}$ are varied. The intersection point of $\Gamma_1$ with the top of the figures is given by $(r, 1)$, where $r \in (0,1)$ solves $r = V(K_1 r + L_1)$. Note that $V(K_1 r + L_1) \to 1$ as $L_1 \to \infty$. It follows that $r \to 1$ as $L_1 \to \infty$.

\subsubsection{``Convex curve disconnected from zero"}

We restrict $\Gamma_1^{K_1, L_1}$ to the second domain of Definition \ref{def:curvedomains}, i.e., $K_1 > 2$, $L_1 > 0$. In this case the curve has the following properties.

\begin{enumerate}
\item By Property 1 of Theorem \ref{thm:propfundcurves} the curve has no turning point.
\item By Property 2 of Theorem \ref{thm:propfundcurves} the curve $\Gamma_1$ is not connected with $(r_1, r_2) = (0,0)$,
\item By Property 3 of Theorem \ref{thm:propfundcurves} $\Gamma_1$ is connected with $(r_1, r_2) = (r,1)$, for some $r \in (0,1)$,
\item There exists an extra (non-trivial) connection with the line $\{ (s, 0) : s \in (0,1) \}$.
\end{enumerate}

To see why the last statement is true, note that the second connection point $(r,0)$ is a solution of $r = V(K_1 r)$. Since $K_1 > 2$, such a point $r \in (0,1)$ exists.

Examples of $\Gamma_1$ in this domain are given in Figures \ref{fig:s23} and \ref{fig:s24} showing how the curve changes as the parameters $K_{1}$ and $L_{1}$ are varied. Similarly to the first domain, the intersection point of $\Gamma_1$ with the top of the figures converges to $1$ as $L_1 \to \infty$. In addition, the same is true as $K_1 \to \infty$.

\subsubsection{``The parabola"}
We restrict $\Gamma_1^{K_1, L_1}$ to the second domain of Definition \ref{def:curvedomains}, i.e., $K_1 > 2$, $L_1 < 0$. In this case the curve has the following properties.

\begin{enumerate}
\item By Property 1 of Theorem \ref{thm:propfundcurves} the curve has a turning point.
\item By Property 2 of Theorem \ref{thm:propfundcurves} the curve $\Gamma_1$ is connected with $(r_1, r_2) = (0,0)$,
\item There exists an extra (non-trivial) connection with the line $\{ (s, 0) : s \in (0,1) \}$.
\end{enumerate}
For the last property, note that the second connection point $(r,0)$ is a solution of $r = V(K_1 r)$ and since $K_1 > 2$, such a point $r \in (0,1)$ exists.
Examples of $\Gamma_1$ in this region are given in Figures  \ref{fig:s43} and \ref{fig:s44} showing how the curve changes as the parameters $K_{1}$ and $L_{1}$ are varied.
    
\begin{figure*}
 \centering
 \begin{subfigure}{0.35\linewidth}
        \centering
        \includegraphics[width=1.0\linewidth]{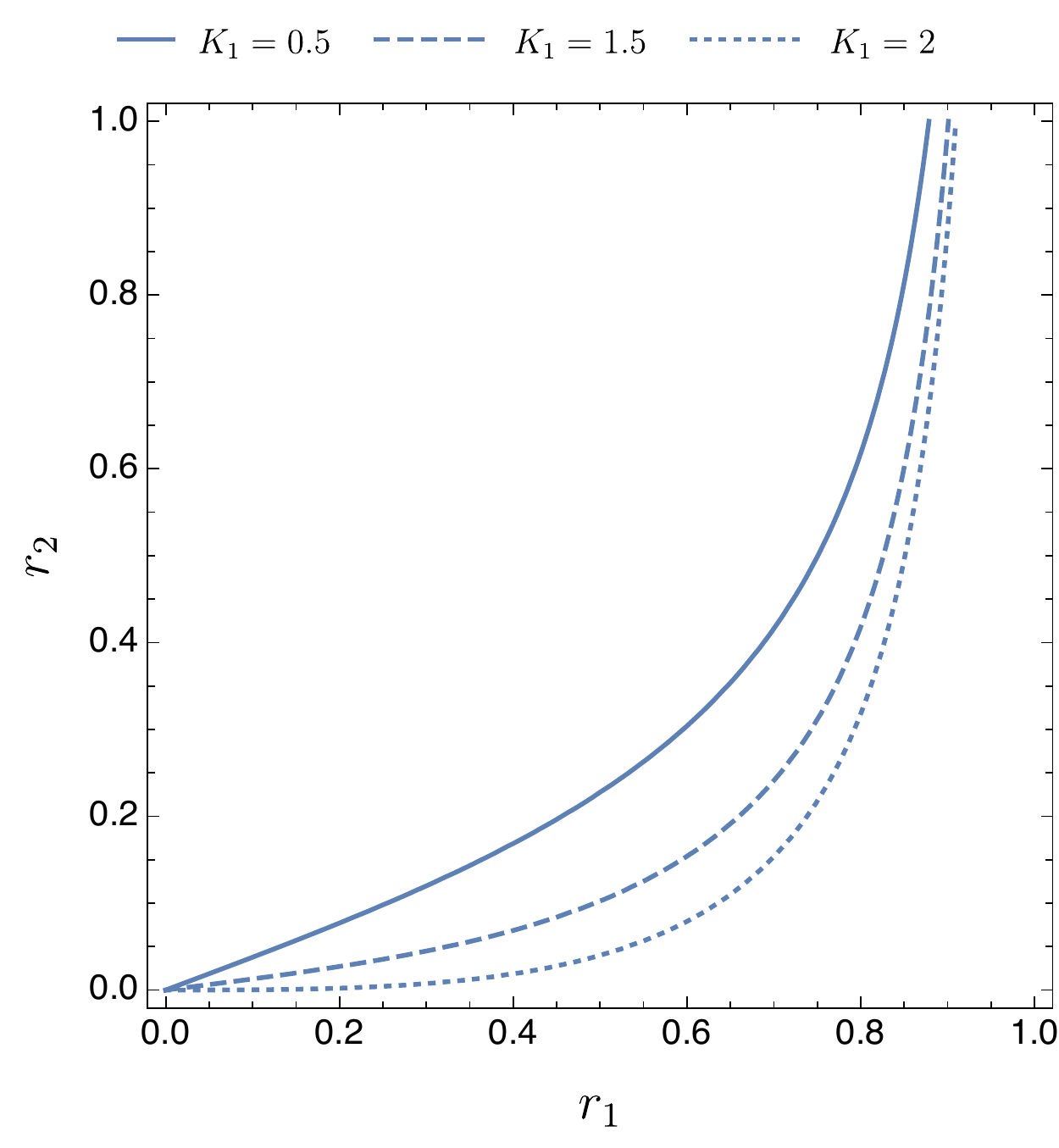} 
        \caption{Plot of $\Gamma_1$ in the first domain with $L_1 = 4$ and $K_1 = 0.5,1.5,2 $.}\label{fig:s13}
    \end{subfigure}\hspace{2.5cm}
    \begin{subfigure}{0.35\linewidth}
        \centering
        \includegraphics[width=1.0\linewidth]{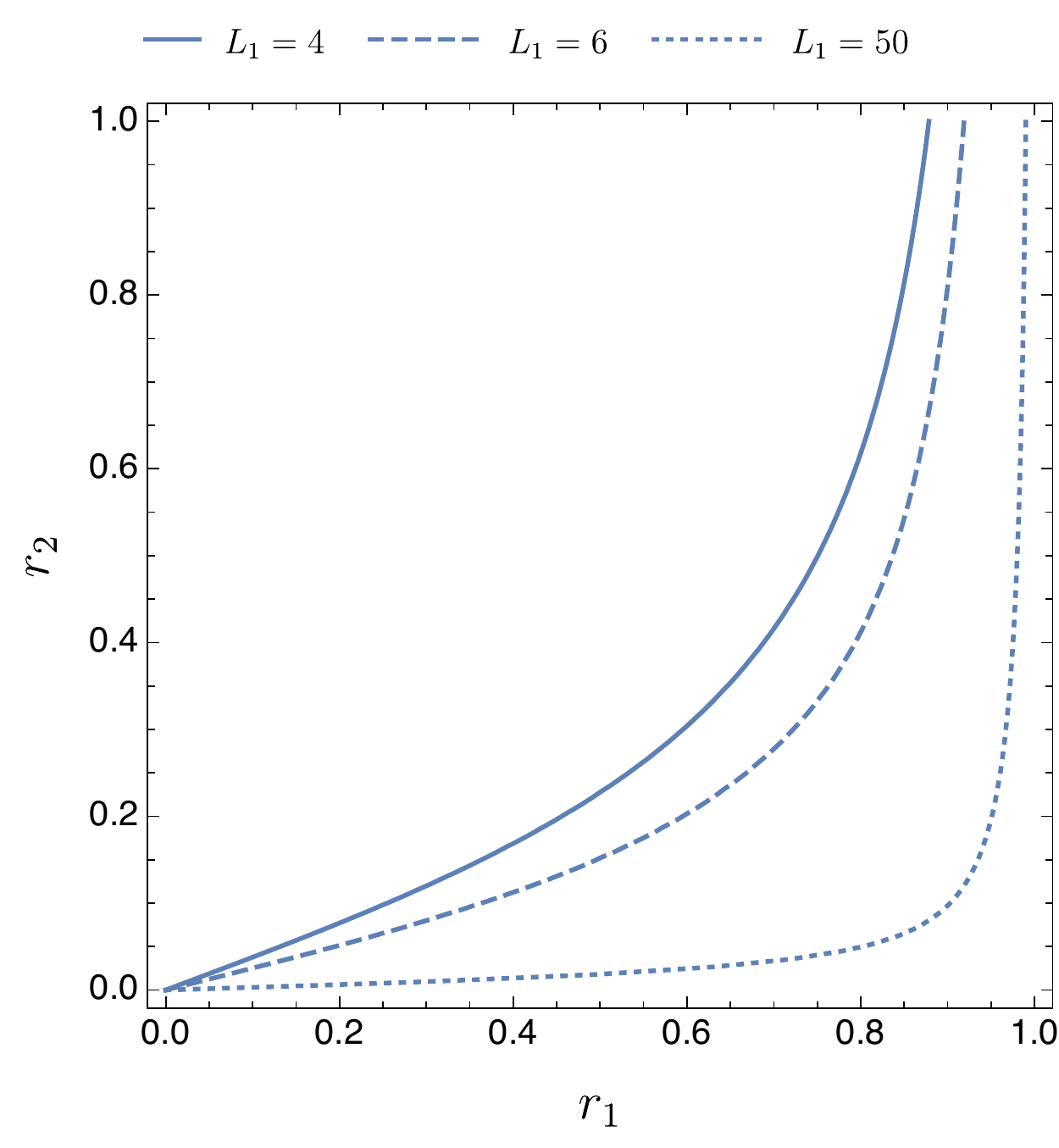} 
        \caption{Plot of $\Gamma_1$ in the first domain with $K_1 = 0.5$ and $L_1 = 4,6,50.$ }\label{fig:s14}
    \end{subfigure}
  \begin{subfigure}{0.35\linewidth}
        \centering
        \includegraphics[width=1.0\linewidth]{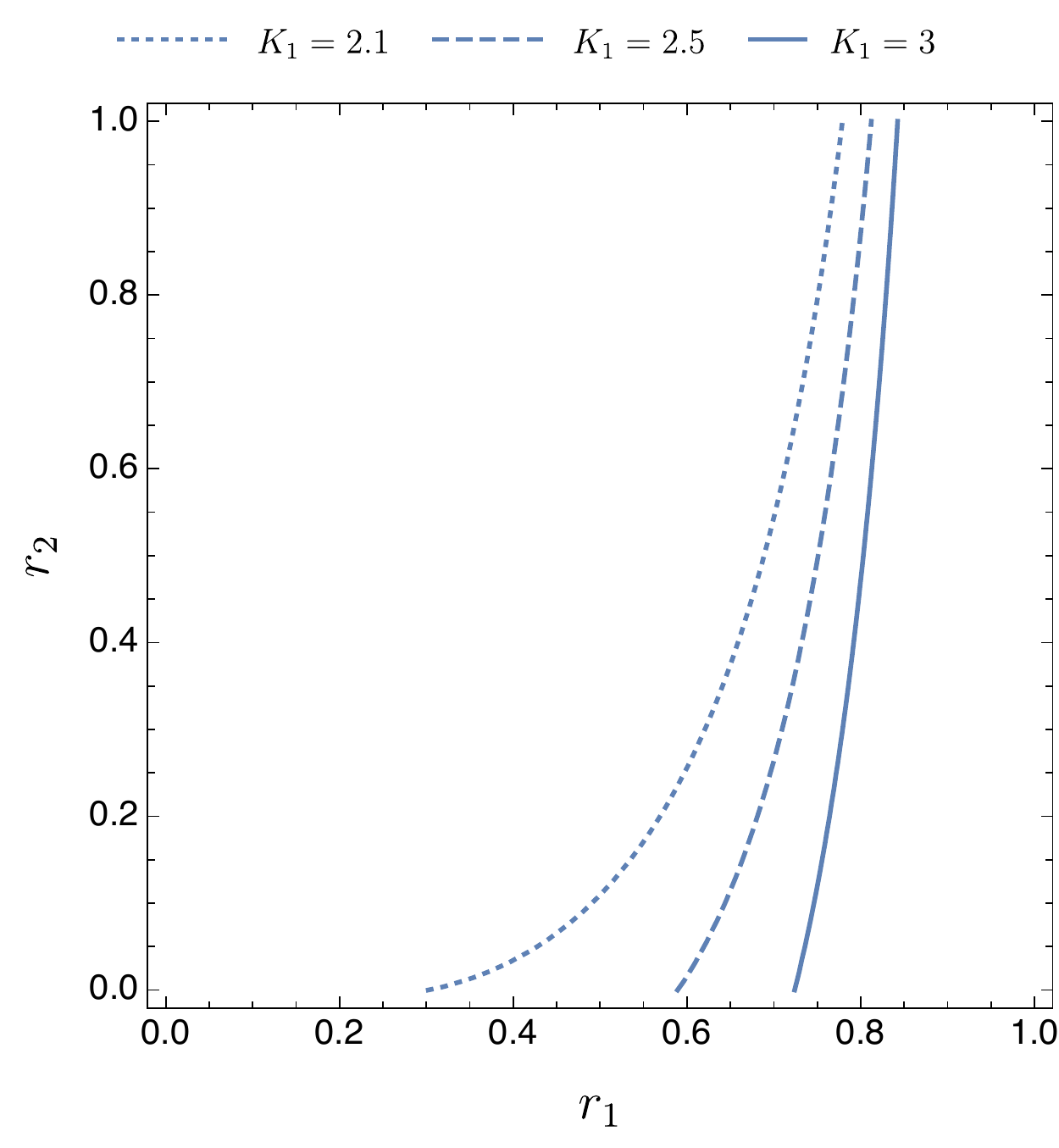} 
        \caption{Plot of $\Gamma_1$ in the second domain with $L_1 = 1$ and $K_1 = 2.1,2.5,3 $.}\label{fig:s23}
    \end{subfigure}\hspace{2.5cm}
    \begin{subfigure}{0.35\linewidth}
        \centering
        \includegraphics[width=1.0\linewidth]{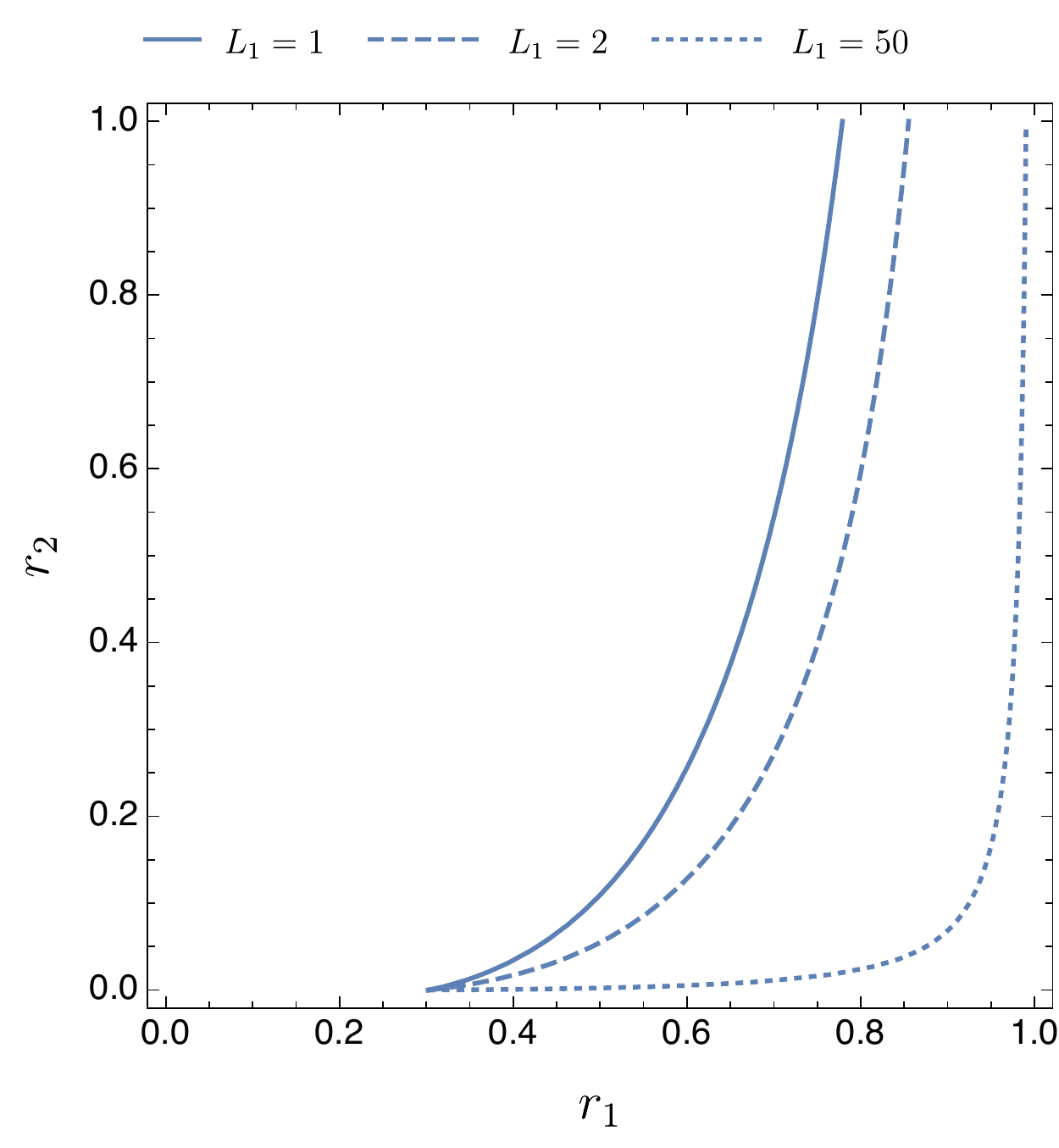} 
        \caption{Plot of $\Gamma_1$ in the second domain with $K_1 = 2.1$ and $L_1 = 1,2,50.$ }\label{fig:s24}
    \end{subfigure}
    \begin{subfigure}{0.35\textwidth}
        \centering
        \includegraphics[width=1.0\textwidth]{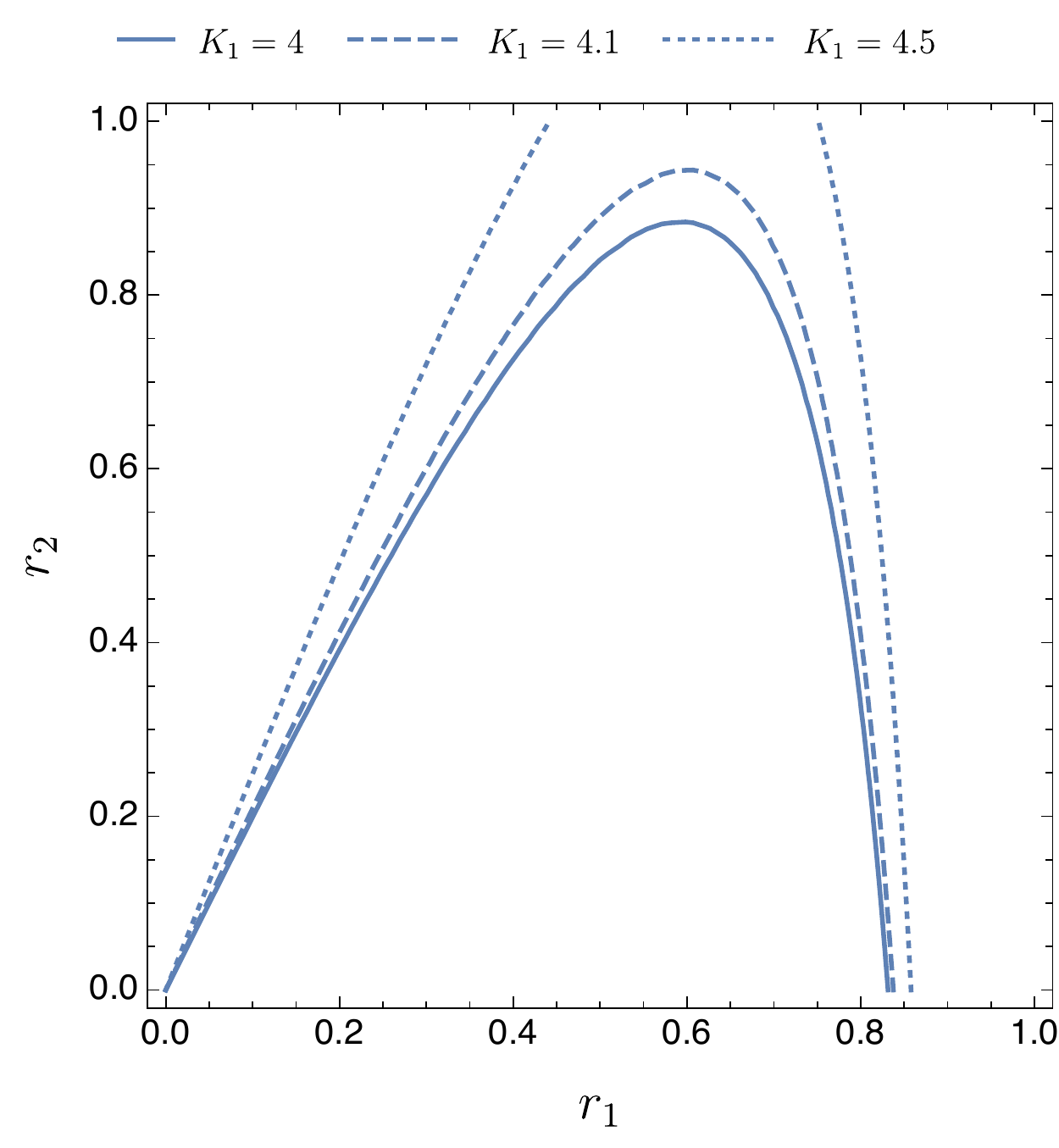} 
        \caption{Plot of $\Gamma_1$ in the third domain with $L_1 = -1$ and $K_1 = 4,4.1,4.5 $.}\label{fig:s43}
    \end{subfigure}\hspace{2.5cm}
    \begin{subfigure}{0.35\textwidth}
        \centering
        \includegraphics[width=1.0\textwidth]{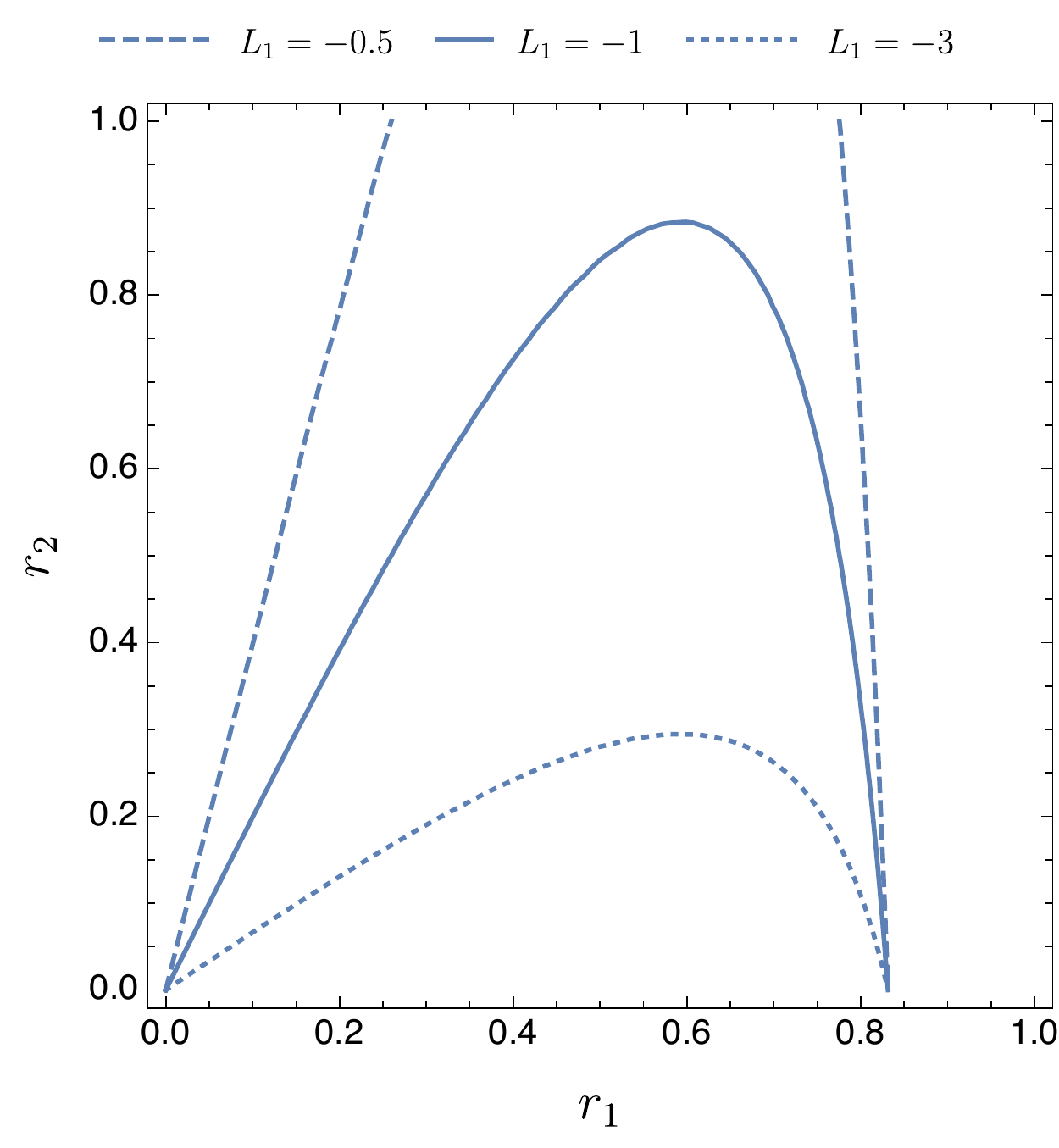} 
        \caption{Plot of $\Gamma_1$ in the third domain with $K_1 = 4$ and $L_1 = -0.5,-1,-3.$ }\label{fig:s44}
    \end{subfigure}
    \caption{Plot of the fundamental curves domains of Definition \ref{def:curvedomains}. In the figures on the left-hand side, $L_1$ is fixed and $K_1$ is varied. In the figures on the right-hand side, $K_1$ is fixed and $L_1$ is varied.}
\end{figure*}

\subsubsection{Preliminary classification regions}
Using the characteristics of the possible fundamental curves we can calculate the cardinality of $\Gamma_1 \cap \Gamma_2$ in the different regions. For $\Gamma_1$ there exist precisely three fundamental curves, which are determined by whether $K_1$ and $K_2$ are either $> 2$ or $\leq 2$ and whether $L_1, L_2$ are either $> 0$ or $< 0$. In addition, there exist precisely three fundamental curves for $\Gamma_2$, hence there are $3^2 = 9$ possible non-trivial regions we need to consider corresponding to all possible combinations of fundamental curves. Using the geometry of the fundamental curves, we can identify the maximum number of solutions possible in each region by the number of intersections possible between the two fundamental curves. These are summarized in Table \ref{fig:overview}. In the following sections we prove results that allow us to refine Table \ref{fig:overview} further.

\begin{table}[!ht]
\centering
\vspace{0.1cm}

\begin{tabular}{l l c} \toprule
		& \textbf{Region}  &\textbf{ Max $\#$ solutions}  \\ \midrule

$\RR_1$	& $K_1 < 2, L_1 < 0 \text{ or } K_2 < 2, L_2  < 0$	&  $1$\\ \hline
$\RR_2$	& $K_1 \leq 2, K_2 \leq 2, L_1 > 0, L_2  > 0$	&  $2$\\ 

$\RR_3$	& $K_1 >  2, K_2 > 2, L_1 > 0, L_2  > 0$		&  $2$\\ 

$\RR_4$	& $K_1 \leq 2, K_2 > 2, L_1 > 0, L_2 > 0$		&  $2$\\ 
$\RR_5$	& $K_1 > 2, K_2 \leq 2, L_1  > 0, L_2  > 0$		&  $2$\\ \hline	
	
$\RR_6$	& $K_1>2, K_2>2, L_1 < 0, L_2> 0$		&  $3$\\ 		
$\RR_7$	& $K_1>2, K_2>2, L_1 > 0, L_2 < 0$		&  $3$\\ 
			
$\RR_8$ & $K_1>2, K_2 \leq 2, L_1  < 0, L_2  > 0$	&  $3$\\ 
$\RR_9$	& $K_1 \leq 2, K_2>2, L_1  > 0, L_2  < 0$	&  $3$\\ \hline

$\RR_{10}$	& $K_1>2, K_2 > 2, L_1 < 0, L_2 < 0$		&  $4$\\ \bottomrule
\end{tabular}
\caption{An overview of all regions in which synchronized solutions can occur and the maximum number of solutions possible (unsycnhronized solutions included). See Figure \ref{numregov} for numerical examples.}\label{fig:overview}
\end{table}

\section{Unsynchronized and symmetrically synchronized solutions}
\label{sec:unsynch}
In this section we first provide sufficient condition for the unsynchronized solution to be the only solution and then analyze the symmetrically synchronized solutions, i.e., solutions where $r_{1}=r_{2}$, which corresponds to the solution in the one-community noisy Kuramoto model.

\subsection{Unsynchronized solutions}
\begin{theorem}[Sufficient condition for the unsynchronized solution]\label{thm:zero}\ \\
Define:
\begin{equation}
\beta^{\zero}(K_1,K_2,L_1,L_2) := (K_1 - 2)(K_2 - 2) - L_1 L_2.
\end{equation}
The unsynchronized solution $(r_1,r_2) = (0,0)$ is the only solu\- tion if the corresponding parameter set of interaction strengths $(K_1, K_2, L_1, L_2)$ is contained in one of the following regions:
\begin{enumerate}
\item $K_1 < 2$, $L_2 > 0$ and $\beta^{\zero} \geq 0$,
\item $K_2 < 2$, $L_1 > 0$ and $\beta^{\zero} \geq 0$,
\item $K_1 > 2$, $L_2 < 0$ and $\beta^{\zero} \leq 0$,
\item $K_2 > 2$, $L_1 < 0$ and $\beta^{\zero} \leq 0$,
\item $K_1 > 2$, $K_2 < 2$, $L_1 < 0$, $L_2 > 0$ and $\beta^{\zero} \geq 0$,
\item $K_1 < 2$, $K_2 > 2$, $L_1 > 0$, $L_2  < 0$ and $\beta^{\zero} \geq 0$,
\item $K_1 \leq 2$ and $L_1 < 0$,
\item $K_2 \leq 2$ and $L_2 < 0$,

\end{enumerate}
\end{theorem}
The proof of Theorem \ref{thm:zero} is given in Appendix \ref{app:zero}.

\begin{remark}
Note that the reverse is not true. If we find a region for which there is no synchronized solution, i.e., $(r_1,r_2) = (0,0)$, then it is possible that \eqref{eq:r11} and \eqref{eq:r22} hold simultaneously. This means that we do not have a full classification of the regions where the only solution is the unsynchronized solution.
\end{remark}

\subsection{Symmetrically synchronized solutions}
We consider a special class of synchronized solutions. We call a synchronized solution $(r_1, r_2)$ a \emph{symmetric synchronized solution} or \emph{symmetric solution} if $r_1 = r_2$. The following two theorems give necessary and sufficient conditions for the existence of symmetric solutions and a characterization of the equivalence class of parameter combinations corresponding to symmetrically synchronized solutions.

\begin{theorem}[Symmetric solution]\label{reg1}
A symmetric solution $(r_1, r_2) = (r,r)$ (for some $r \in (0,1)$) exists if and only if
\begin{equation}
\begin{cases}
K_1 + L_1 &> \quad 2,\\
K_1 + L_1 &= \quad K_2 + L_2. \label{eq:symsol}
\end{cases}
\end{equation}

\end{theorem}

\begin{proof}
Suppose $r_1 = r_2 = r$, for some $r \in (0,1)$, then equations \eqref{eq:scp1} reduce to:
\begin{equation}
r	= V([K_1 + L_1] r),\quad \text{and}\quad r = V([K_2 + L_2] r).
\end{equation}
We know that $V(K r) = r$ has a solution $r > 0$ if and only if $K > 2$. Hence it follows that $K_1 + L_1 > 2$ and $K_2 + L_2 > 2$. In order to have $V([K_1 + L_1] r) = V([K_2 + L_2] r)$ it is necessary and sufficient for  $K_1 + L_1 = K_2 + L_2$ due to the fact that $V$ is increasing and concave.
\end{proof}

\begin{remark}
Suppose we have two model parameter vectors $(K_1,K_2, L_1, L_2)$, $(K_1',K_2',L_1',L_2') \in \R^4$ satisfying \eqref{eq:symsol}. The relation
\begin{equation}
(K_1, K_2, L_1, L_2) \sim_S (K'_1, K'_2, L'_1, L'_2) \iff K_1 + L_1 = K'_1 + L'_1  
\end{equation}
is an equivalence relation. If the vector $(K_1, K_2, L_1, L_2) \sim_S (K'_1, K'_2, L'_1, L'_2)$ and $K_1 + L_1 = C$ then both vectors are in the same \emph{symmetry class} of \emph{level} $C$, denoted by $[C]$. The set of all symmetry classes is denoted by $\R^4/ \sim_S$.
\end{remark}

The following result states that a symmetry class uniquely determines the synchronization levels of the corresponding symmetric solutions.

\begin{theorem}[Uniqueness of symmetric solution up to symmetry class]\ \\
Suppose we have two parameter vectors $(K_1, K_2,L_1, L_2)$, $(K'_1, K'_2,L'_1, L'_2) \in \R^4$ with symmetric solution $(r,r)$ and $(r',r')$, then the following are equivalent:
\begin{enumerate}
\item $r = r'$,
\item $(K_1, K_2,L_1, L_2),(K'_1, K'_2,L'_1, L'_2) \in [C]$, with $C = K_1 + L_1$.
\end{enumerate}
It follows that each symmetry class $[C]$ is uniquely determined by the corresponding synchronization level $r(C)$, which is the solution of $r = V(C r)$. In other words there is a one-to-one correspondence between the symmetry class $[C]$ and the synchronization level $r(C)$.
\end{theorem}

\begin{proof}
We will prove both implications.\\
 $(1 \Rightarrow 2)$: Assume that $r = r'$, then 
\begin{equation}
r =  V([K_1 + L_1  ] r) = V([K'_1 + L'_1]r ),
\end{equation}
which implies $K_1 + L_1 > 0$ and $K'_1 + L'_1 > 0$, since $r \in (0,1)$. Now since $V$ is strictly increasing on $(0, \infty)$ we have $K_1 + L_1 = K'_1 + L'_1 $, hence 
\begin{equation}
(K_1, K_2, L_1, L_2) \sim_S (K'_1, K'_2, L'_1, L'_2).
\end{equation}
Take $C = K_1 + L_1$, then $C \in \R_{> 2}$ because we assumed that a symmetric solution exists. The result now follows. \newline \\
$(2 \Rightarrow 1):$ We assume that $(K_1, K_2, L_1, L_2) \sim_S (K'_1, K'_2, L'_1, L'_2)$ and $K_1 + L_1 = C$, with $C \in \R_{> 2}$. Furthermore we have 
\begin{align}
r  &= V( [K_1 + L_1 ] r ),\\
r' &= V( [K'_1 + L'_1] r' ).
\end{align}
Since both parameter vectors are in the same symmetry class we have $K_1  + L_1 = K_2 + L_2$ and therefore $r = r'$.
\newline \\
To conclude, note that if $(K_1, K_2, L_1, L_2) \in [C]$, then 
\begin{equation}
K_1 + L_1 = K_2 + L_2 = C,
\end{equation}
and therefore the corresponding synchronization level $r(C)$ is the solution of 
\begin{equation}
r = V([K_1 + L_1] r ) = V(C r).
\end{equation}
This synchronization level is uniquely determined because $C \mapsto r(C)$ is strictly increasing.
\end{proof}

Properties and asymptotic expressions for $C \mapsto r(C)$ are given in Proposition \ref{prop:r(C)} and Proposition \ref{prop:asympr(C)} in the appendix.

\section{Conclusion}
\label{sec:conclusion}
We have introduced a geometric interpretation of the self-consistency equations for the two-community noisy Kuramoto model that allow us to analyze when and how many solutions to the self-consistency equations exist. We have also analyzed two types of solutions more explicitly, namely, the unsynchronized and symmetrically synchronized solutions and have shown that the phase difference between the average phases of the two communities is always zero or $\pi$ in steady-state.

In a second paper of this series we will make use of the geometric interpretation to further refine the ten regions identified here. This refinement will rely on the identification of all possible bifurcation points arising in this model. Each type of bifurcation point gives rise to a solution boundary in the phase diagram separating regions with a different number of solutions. Furthermore, the geometric interpretation allows us to easily calculate the asymptotes of these solution boundaries.

In terms of applications, it is interesting to note that the synchronization levels in the two communities of the SCN are typically thought to be symmetrically synchronized \cite{Buijink2016}. Our analysis of the symmetrically synchronized solution could then be applied to estimate the interaction strength and noise strength parameters when modeling the SCN by a two-community Kuramoto model. Since our analysis also identifies critical points in the phase diagram, such an estimation could shed more light on the critical brain hypothesis \cite{Chialvo2010}.

\begin{acknowledgments}
The authors are grateful to F.\ den Hollander for guiding discussions and detailed comments.
\end{acknowledgments}

\section*{Data Availability}
Data sharing is not applicable to this article as no new data were created or analyzed in this study.

\appendix

\section{Proofs omitted from the main text}
\subsection{Proof of Proposition \ref{prop:selfcons}}
\label{app:selfcons}
\begin{proof}
First we note that 
\begin{equation}
\int_\S \eee^{a \cos \theta + b \sin \theta} \ddd \theta = 2 \pi \I_0 ( \sqrt{a^2 + b^2} ), \label{eq:B1 }
\end{equation}
due to the identity in \cite[pg. 339 equation 3.338 4]{Jeffrey2007} with $a=0$, $p=0$, $q=0$. Differentiating the left and right-hand side of (\ref{eq:B1 }) with respect to $a$ and $b$ gives 
\begin{align}
\int_\S \cos \theta \eee^{a \cos \theta + b \sin \theta} \ddd \theta  =  2 \pi \frac{\p }{\p a}  \I_0 ( \sqrt{a^2 + b^2} ),\\
\int_\S \sin \theta \eee^{a \cos \theta + b \sin \theta} \ddd \theta  =  2 \pi \frac{\p }{\p b}   \I_0 ( \sqrt{a^2 + b^2} )\nonumber.
\end{align}
By using the identity in \cite[9.6.27]{Abramowitz65} we obtain 
\begin{align}
\label{eq:besselidentityderivative}
\int_\S \cos \theta \eee^{a \cos \theta + b \sin \theta} \ddd \theta = \frac{2 \pi a \I_1(\sqrt{a^2 + b^2})}{\sqrt{a^2 + b^2}},\\
\int_\S \sin \theta \eee^{a \cos \theta + b \sin \theta} \ddd \theta = \frac{2 \pi b \I_1(\sqrt{a^2 + b^2})}{\sqrt{a^2 + b^2}},\nonumber
\end{align}
with $a, b, \in \R$ and $\I_m(x) := \frac{1}{2 \pi} \int_\S (\cos \theta)^m \exp(x \cos \theta) \ddd \theta$ the \emph{modified Bessel function of the first kind}, $m \in \{0,1\}$.
The trigonometric identity $\cos(a - b) = \cos a\cos b + \sin a \sin b$, for $a, b \in \R$, allows us to write
\begin{align}
 L_1 r_2 \cos(\psi_2 - \theta) + K_1 r_1 \cos(\psi_1 - \theta)	 &= a_1 \cos \theta + b_1 \sin \theta,
\end{align}
with 
\begin{align}
a_1 &= L_1 r_2 \cos \psi_2 + K_1 r_1 \cos \psi_1,\\
b_1 &= L_1 r_2 \sin \psi_2 + K_1 r_1 \sin \psi_1,
\end{align}
which allows us, together with \eqref{eq:besselidentityderivative}, to rewrite \eqref{eq:osc1} as
\begin{equation}
r_k = \frac{( a_k \cos \psi_k + b_k \sin \psi_k ) \I_1\left( \sqrt{a_k^2 + b_k^2} \right)}{ \sqrt{a_k^2 + b_k^2} \I_0 \left( \sqrt{a_k^2 + b_k^2}\right)}, 
\end{equation}
for $k\in \{1, 2\}$. For the equation with $k=2$ we have
\begin{align}
a_2 &= L_2 r_1 \cos \psi_1 + K_2 r_2 \cos \psi_2,\\
b_2 &= L_2 r_1 \sin \psi_1 + K_2 r_2 \sin \psi_2.
\end{align}
We define $\psi := \psi_2 - \psi_1$ and note that
\begin{align}
 a_1 \cos \psi_1 + b_1 \sin \psi_1 &= K_1 r_1 + L_1 r_2 \cos \psi,\\
 a_2 \cos \psi_2 + b_2 \sin \psi_2	&= K_2 r_2 + L_2 r_1 \cos \psi.
\end{align}
In addition, we have 
\begin{align}
a_1^2 + b_1^2	&= K_1^2 r_1^2 + L_1^2 r_2^2 + 2 K_1 L_1 r_1 r_2 \cos \psi,\\
a_2^2 + b_2^2 	&= K_2^2 r_2^2 + L_2^2 r_1^2 + 2 K_2 L_2 r_1 r_2 \cos \psi.
\end{align}

Since $V(x) = \frac{\I_1(x)}{\I_0(x)}$ and $W(x) = \frac{2V(x)}{x}$, we have by substitution

\begin{align}
r_k = &\frac{K_k r_k + L_k r_{k'} \cos \psi}{2} \\
&\times W \left( \sqrt{ K_k^2 r_k^2 + L_k^2 r_{k'}^2 + 2 K_k L_k r_k r_{k'} \cos \psi} \right)\nonumber
\end{align}
For the second part note that $\sin(a-b) = \sin a \cos b - \cos a \sin b$, for $a, b \in \R$. We can now rewrite \eqref{eq:osc3}
\begin{equation}
0 = \frac{(a_k \sin \psi_k - b_k \cos \psi_k)\I_1\left( \sqrt{a_k^2 + b_k^2} \right)}{ \sqrt{a_k^2 + b_k^2} \I_0 \left( \sqrt{a_k^2 + b_k^2}\right)}.
\end{equation}
Note that 
\begin{align}
a_1 \sin \psi_1 - b_1 \cos \psi_1	&= - L_1 r_2 \sin( \psi),\\
a_2 \sin \psi_2 - b_2 \cos \psi_2	&= - L_2 r_1 \sin(\psi),
\end{align}
with $\psi = \psi_2 - \psi_1$. Substitution and multiplying both sides with $-2$ gives equations \eqref{eq:sc3}.
\end{proof}

\subsection{Proof for Theorem \ref{thm:hneg}}
\label{app:proofhneg}
\begin{proof}
We will prove both implications.\\
($``\Rightarrow"$) Note that $h_1^{K_1, L_1}(r_1, r_2) < 0$ implies that
\begin{equation}
V(K_1 r_1 + L_1 r_2 ) < r_1. \label{eq:ine2}
\end{equation}
Since this inequality holds true for all $(r_1, r_2) \in [0,1]^2\setminus \{(0,0)\}$ we consider points on the edge of the unit square. Let $(r_1, r_2) = (r , 0)$ for some $r \in (0,1)$, then (\ref{eq:ine2}) reduces to 
\begin{equation}
V(K_1 r  ) < r. \label{eq:ine3} 
\end{equation}
Note (\ref{eq:ine3}) is true for all $r \in (0,1)$ and therefore $K_1 < 2$. Now take $(r_1, r_2) = (0 , r)$ for some $r \in (0,1)$, then (\ref{eq:ine2}) reduces to 
\begin{equation}
V( L_1 r  ) < 0,
\end{equation}
which implies that $L_1 < 0$.\\
($``\Leftarrow"$)
Let $(r_1,r_2) \in [0,1]^2\setminus \{(0,0)\}$, then
\begin{equation}
\frac{K_1 r_1}{2} + \frac{L_1 r_2}{2} \leq r_1 + \frac{L_1 }{2} r_2 \leq r_1.
\end{equation}
Here we used that $\frac{K_1 r_1}{2} \leq r_1$, $\frac{L_1}{2} < 0$ and $r_2 \geq 0$. In order to conclude the proof we make use of the following lemma.

\begin{lemma}\label{thm:neg}
Let $(r_1, r_2) \in [0,1]^2\setminus \{(0,0)\}$, $K_1 \in \R$ and $L_1 \in \R \setminus \{0\}$. If
\begin{equation}
r_1 \geq \frac{K_1 r_1 + L_1 r_2}{2},
\end{equation}
then $h_1^{K_1, L_1}(r_1, r_2) < 0$.
\end{lemma}
\begin{proof}
We will prove the Lemma by contradiction. Suppose that $h_1^{K_1, L_1}(r_1, r_2) \geq 0$. Then
\begin{equation}
r_1 \leq V(K_1 r_1 + L_1 r_2 ).\label{eq:hp1}
\end{equation}
Since $r_1 \geq 0$, we have $V(K_1 r_1 + L_1 r_2  ) > 0$. Note that $V(K_1 r_1 + L_1 r_2  ) = 0$ is not possible, because this is true if and only if $K_1 r_1 + L_1 r_2 = 0$, while we assumed that $(r_1, r_2) \neq (0,0)$ and $L_1 \neq 0$. It follows that $K_1 r_1 + L_1(\cos \psi) r_2  > 0$. We use property $6$ in \ref{lV} to get
\begin{equation}
V(K_1 r_1 + L_1 r_2  ) < \frac{K_1 r_1 + L_1 r_2 }{2}. \label{eq:hp2}
\end{equation}
Combining (\ref{eq:hp1}) and (\ref{eq:hp2}), we get
\begin{equation}
r_1 < \frac{K_1 r_1 + L_1 r_2 }{2}, 
\end{equation}
from which the claim follows.
\end{proof}
This concludes the proof.
\end{proof}

\subsection{Proof of Lemma \ref{lem:pderivatives}}
\label{app:pderivatives}

\begin{proof}
We prove the four claims separately.
\begin{enumerate}

\item The curve $\Gamma_1$ is described by the equation $V(K_1 r_1 + L_1 r_2(r_1)) - r_1 = 0$, where $r_2$ depends on $r_1$. Differentiating with respect to $r_1$, we obtain

\begin{equation}
\left(K_1 + L_1 \frac{\p \Gamma_1}{\p r_1} \right) C_{1,1} - 1 = 0, \label{eq:dife1}
\end{equation}
which gives
\begin{equation}
\frac{\p \Gamma_1}{\p r_1} = \frac{1 - K_1 C_{1,1}}{L_1 C_{1,1}}. \label{eq:dife2}
\end{equation}

\item Differentiating \eqref{eq:dife1} with respect to $r_1$, we obtain
\begin{align}
K_1 &\left( K_1 + L_1 \frac{\p \Gamma_1}{\p r_1} \right) C_{1,2} + L_1 C_{1,1} \frac{\p^2 \Gamma_1}{\p r_1^2} \nonumber\\
&+ L_1 \frac{\p \Gamma_1}{\p r_1} \left( K_1 + L_1 \frac{\p \Gamma_1}{\p r_1} \right) C_{1,2} = 0. \label{eq:dife22}
\end{align} 

Using \eqref{eq:dife2}, we find 
\begin{equation}
K_1 + L_1 \frac{\p \Gamma_1}{\p r_1} = \frac{1}{C_{1,1}}. \label{eq:dife3}
\end{equation}
Substituting \eqref{eq:dife2} and \eqref{eq:dife3} into \eqref{eq:dife22} and multiplying by $C_{1,1}$, we find
\begin{equation}
L_1 C_{1,1}^2 \frac{\p^2 \Gamma_1}{\p r_1^2} + \frac{C_{1,2}}{C_{1,1}} = 0.
\end{equation}
Rewriting gives the desired result.
\item The curve $\Gamma_2$ is described by the equation $V(K_2 r_2(r_1) + L_2 r_1) - r_2(r_1) = 0$. Differentiating with respect to $r_1$, we obtain
\begin{equation}
\left( K_2 \frac{\p \Gamma_2}{\p r_1} + L_2 \right) C_2 - \frac{\p \Gamma_2}{\p r_1} = 0, \label{eq:dife4}
\end{equation}
which implies
\begin{equation}
\frac{\p \Gamma_2}{\p r_1} = \frac{L_2 C_2}{1 - K_2 C_2}.
\end{equation}

\item Differentiating \eqref{eq:dife4} with respect to $r_1$, we obtain

\begin{align}
K_2 &C_{2,1}  \frac{\p^2 \Gamma_2}{\p r_1^2} + K_2 \frac{\p \Gamma_2}{\p r_1} \left( K_2 \frac{\p \Gamma_2}{\p r_1} + L_2 \right) C_{2,2} \nonumber\\
&+ L_2 \left( K_2 \frac{\p \Gamma_2}{\p r_1} + L_2 \right) C_{2,2} - \frac{\p^2 \Gamma_2}{\p r_1^2} = 0,
\end{align}
which can be written as

\begin{align}
(1 - K_2 C_{2,1} ) \frac{\p^2 \Gamma_2}{\p r_1^2} = &\Big[ K_2 \frac{\p \Gamma_2}{\p r_1} \left( K_2 \frac{\p \Gamma_2}{\p r_1} + L_2 \right) \nonumber\\
&+ L_2 \left( K_2 \frac{\p \Gamma_2}{\p r_1} + L_2 \right) \Big] C_{2,2}. \label{eq:dife5}
\end{align}

Using \eqref{eq:dife4}, we find 
\begin{equation}
K_2 \frac{\p \Gamma_2}{\p r_1} + L_2 = \frac{L_2}{1 - K_2 C_{2,1}}. \label{eq:dife6}
\end{equation}
Substituting \eqref{eq:dife6} into \eqref{eq:dife5} gives the claim.
\end{enumerate}
\end{proof}

\subsection{Proof of Theorem \ref{thm:propfundcurves}}
\label{app:propfundcurves}
\begin{proof}\ \\
\textbf{Property 1:} Take $\psi = 0$. We have 
\begin{equation}
\frac{\p \Gamma_1}{\p r_1} = 0 \iff C_{1,1} = \frac{1}{K_1}. \label{eq:ct1}
\end{equation}
We assume that the turning point can occur at a point $(r_1, r_2)$, with $r_1 \in (0,1)$ and $r_2 > 0$. Furthermore, since $(r_1, r_2) \in \Gamma_1$ we have $r_1 = V(K_1 r_1 + L_1 r_2)$. Hence, it follows from Lemma \ref{derV1} that
\begin{equation}
C_{1,1} = 1 - \frac{r_1}{K_1 r_1 + L_1 r_2} - r_2^2.
\end{equation}
Solving $C_{1,1} = \frac{1}{K_1}$ in terms of $r_2$, we find
\begin{equation}
r_2(r_1) = \frac{-K_1^2 r_1^3+K_1^2 r_1-2 K_1 r_1}{L_1 \left(K_1 r_1^2-K_1+1\right)}. \label{eq:er2}
\end{equation}
The existence of a turning point requires that there exists a $r_1 \in (0,1)$ such that $r_2(r_1) = 0$. Solving for $r_2 = 0$, we get three solutions:
\begin{equation}
r_1 = 0, \quad r_1 = \sqrt{1 - \frac{2}{K_1} }, \quad r_1 = - \sqrt{ 1 - \frac{2}{K_1} }.
\end{equation}
Since $r_1 \in [0,1]$, we can discard the last solution. Furthermore, for $K_1 < 2$ the second solution is complex and at $K_1 =2$ the second solution is zero. Hence we may assume that $K_1 > 2$.

It is easy to see that $r_2$ is mirrored along the $x$-axis as $L_1 \mapsto - L_1$. Furthermore, if $L_1 < 0$, then $r_2(r_1) \geq 0$ when $0 \leq r_1 \leq \sqrt{1 - \tfrac{2}{K_1}}$ and $r_2(r_1) > 0$ when $\sqrt{1 - \tfrac{2}{K_1}} < r_1 \leq 1$. Since a solution requires that $(r_1, r_2) \in (0,1)^2$, it follows that if $L_1 >0$, then the possible solutions are
\begin{equation}
\left\{ (r_1, r_2) \in (0,1)^2 : \sqrt{1 - \tfrac{2}{K_1}} < r_1 < 1,~ r_2 > 0 \right\}. \label{eq:solset1}
\end{equation}

We next show that every solution in the solution set in \eqref{eq:solset1} is not contained in $\Gamma_1$. Using \eqref{eq:er2}, we find
\begin{align}
K_1 r_1 + L_1 r_2 &= K_1 r_1 + \frac{-K_1^2 r_1^3+K_1^2 r_1-2 K_1 r_1}{K_1 r_1^2-K_1+1} \nonumber\\
&= \frac{ K_1 r_1}{K_1(1 - r_1^2) - 1}.
\end{align}
It follows that a turning point solution on the curve $\Gamma_1$ solves 

\begin{equation}
r_1 = V\left( \frac{ K_1 r_1}{K_1(1 - r_1^2) - 1} \right). \label{eq:sysv1}
\end{equation}
Note that $r_1 \mapsto \frac{ K_1 r_1}{K_1(1 - r_1^2) - 1}$ is strictly increasing, and let $r_1(K_1)$ be the non-trivial solution of \eqref{eq:sysv1}. By plotting $r_1(K_1)$ and $\sqrt{1 - \tfrac{2}{K_1}}$, we see that $r_1(K_1) < \sqrt{1 - \tfrac{2}{K_1}}$ for all $K_1 > 2$. Hence, all points in the solution set given in \eqref{eq:solset1} are not on the curve $\Gamma_1$. If $L_1 < 0$, then a solution $(r_1, r_2) \in \Gamma_1$ satisfying \eqref{eq:ct1} exists and is unique. Furthermore, $\Gamma_1$ is strictly concave for $L_1 < 0$ so that this solution is a turning point.\\
\textbf{Property 2:} We prove this property by Taylor expanding $h_1(r1, r2)$ around $(r_1, r_2) = (0, 0)$ and showing that it is only possible to solve the equation $h_1(\ee, \dd) = 0$ in the domains mentioned in the theorem. The Taylor expansion leads to

\begin{equation}
h_1(\ee, \dd) = \left( \frac{K_1}{2} - 1 \right) \ee + \frac{L_1}{2} \dd + O(\ee^2) + O(\dd^2). \label{eq:ccond}
\end{equation}
Setting this to zero and solving for $\ee$, we get
\begin{equation}
\ee = - \frac{L_1}{(K_1 - 2)} \dd.
\end{equation}
Since we require both $\ee$ and $\dd$ to be larger than zero, we see that $(\ref{eq:ccond})$ can only hold in the domains of the theorem.\\
\textbf{Property 3:} Note that the function $r \mapsto V(K_1 r + L_1)$ is continuous on $[0,1]$, is starting at $V(L_1) > 0$ and is bounded by $1$. Hence there exists some $r \in (0,1)$ such that $V(K_1 r + L_1) - r = 0$. The uniqueness of this $r \in (0,1)$ follows from the observation that $\Gamma_1$ is strictly concave. The claim follows because $\Gamma_1$ is continuous on $[0,1]\setminus\{0,0\}$ by Lemma \ref{lem:ccurve}.
\end{proof}

\subsection{Proof of Theorem \ref{thm:zero}}
\label{app:zero}
\begin{proof}
First note that by Lemma \ref{0r} solutions of the form $r_1 = 0$ and $r_2 > 0$, and vice-versa, do not exist. Suppose we have a strictly positive solution $(r_1, r_2)$, i.e.
\begin{equation}
r_1	= V(K_1 r_1 + L_1 r_2 ),\quad \text{and}\quad r_2	= V(K_2 r_2 + L_2 r_1 ),
\end{equation}
with $r_1, r_2 > 0$. Since $V(x) > 0$ if and only if $x > 0$. It follows that
\begin{equation}
K_1 r_1 + L_1 r_2  > 0,\quad \text{and}\quad K_2 r_2 + L_2 r_1  > 0.
\end{equation}
Now we use property $6$ in Lemma \ref{lV} to get
\begin{align}
r_1 &< \frac{K_1 r_1 + L_1 r_2 }{2}, \label{eq:r11}\\
r_2 &< \frac{K_2 r_2 + L_2 r_1 }{2}. \label{eq:r22}
\end{align}

Hence if for given $(K_1,K_2,L_1,L_2) \in \R^4$ there is a positive solution, then the equation (\ref{eq:r11}) and (\ref{eq:r22}) hold simultaneously. Now we proceed as follows: we will prove the contrapositive. We will show that for the regions stated in the theorem equations \eqref{eq:r11} and \eqref{eq:r22} do not hold simultaneously.
\begin{itemize}
\item \textbf{Region 1 \& 2:} 
First, assume that $K_1 < 2$ and $L_2  > 0$. Rewriting equation \eqref{eq:r11} and using $K_1 < 2$ gives
\begin{equation}
r_1 \left(1 - \frac{K_1}{2} \right) < \frac{L_1 r_2}{2}.
\end{equation}
Since $K_1 < 2$,
\begin{equation}
r_1 < \frac{L_1}{2 - K_1} r_2. \label{eq:r12}
\end{equation}
Now if $L_2 > 0$, then substitution of \eqref{eq:r12} in \eqref{eq:r22} gives
\begin{align}
r_2 < \frac{K_2 r_2}{2} + \frac{L_2}{2} \frac{L_1}{2 - K_1} r_2  = \frac{1}{2} r_2 \left[ K_2  + \frac{L_1 L_2}{2 - K_1} \right].
\end{align}
Multiplying both sides by $2/r_2$ and rewriting gives
\begin{equation}
\beta^{\zero}(K_1, K_2, L_1, L_2) = (K_1 - 2)(K_2 - 2) - L_1 L_2  < 0.
\end{equation}
We conclude that, in order to have no synchronized solutions we need to have $\beta^{\zero} \geq 0$.
The same can be done when $K_2 < 2$ and $L_2 > 0$, by first rewriting equation \eqref{eq:r22}.

\item \textbf{Region 3 \& 4:} Assume $K_1 > 2$ and $L_2 \cos \psi < 0$. Using $K_1 > 2$, we obtain by rewriting \eqref{eq:r11} that
\begin{equation}
r_1 > \frac{L_1 r_2}{2 - K_1},
\end{equation}
Multiplying both sides by $-\frac{L_2 }{2}$ gives
\begin{equation}
-\frac{L_2 r_1}{2} > - \frac{1}{2} \frac{L_1 L_2}{2 - K_1} r_2, \label{eq:-r12}
\end{equation}
since $L_{2}<0$. Substituting \eqref{eq:-r12} in \eqref{eq:r22} and rearranging gives
\begin{equation}
\beta^{\zero}(K_1,K_2,L_1,L_2) = (K_1 - 2)(K_2 - 2) - L_1 L_2  > 0.
\end{equation}
In order to have no synchronized solutions, we need to have $\beta^{\zero} \leq 0$. We can apply the same procedure as above for the parameter values $K_2 > 2$ and $L_1 < 0$.

\item \textbf{Region 5 \& 6:}
Let us first consider region $5$. Using $K_1 > 2$ and $K_2 < 2$ rewriting \eqref{eq:r11} and \eqref{eq:r22} gives
\begin{equation}
r_1 > \frac{L_1}{2 - K_1} r_2 \quad \text{and} \quad r_2 < \frac{L_2}{2 - K_2} r_1. \label{eq:las1}
\end{equation}
Note that 
\begin{equation}
\frac{L_2}{2 - K_2} > 0 \quad \text{and} \quad \frac{L_1}{2 - K_1} < 0. \label{eq:las2}
\end{equation}
Combining \eqref{eq:las1} and \eqref{eq:las2} gives
\begin{equation}
r_1 > \frac{L_1 L_2}{(K_1 - 2)(K_2 - 2)} r_1,
\end{equation}
from which it follows that $\beta^{\zero} < 0$. Hence in the case that $\beta^{\zero} \geq 0$ equations \eqref{eq:r11} and \eqref{eq:r22} do not hold simultaneously. The proof for region $6$ follows by a similar argument.

\item \textbf{Region 7 \& 8:} This is a direct consequence of Theorem \ref{thm:hneg}.

\end{itemize}
Since we have shown that for each region in the theorem the equations \eqref{eq:r11} and \eqref{eq:r22} do not hold simultaneously the proof is complete.
\end{proof}

\section{Properties of V and W}

\begin{lemma}[Properties of $V$] The function $V$ has the following properties: \label{lV}
\begin{enumerate}
\item $V(0) = 0$,
\item $V'(0) = \frac{1}{2}$,
\item $\lim_{x \to \infty} V(x) = 1$,
\item $V'(x)>0$ everywhere except at zero,
\item $V''(x)<0$ everywhere except at zero,
\item $V(x) < \frac{x}{2}$ for $x \in (0,\infty)$,
\item $V(-x) = - V(x)$ for all $x \in (0,\infty)$.
\end{enumerate}
\end{lemma}

\begin{proof}
The proof is given in \cite[Proposition 3.1]{Meylahn2020}
\end{proof}

\begin{lemma}[Derivatives of $V$] \label{derV1}
The first and second derivative of $V$ are given by
\begin{align}
V'(x) &= - V(x)^2 - \frac{V(x)}{x} + 1,\\
V''(x)&= 2 V(x)^3 + \frac{3}{x} V(x)^2 + \left(\frac{2}{x^2} - 2\right) V(x) - \frac{1}{x}.
\end{align}
\end{lemma}
\begin{proof}
Note that

\begin{align}
V'(x) &= \frac{\int_\S \eee^{x \cos \theta } \ddd \theta \int_\S \cos^2 \theta \eee^{x \cos \theta} \ddd \theta}{(\int_\S \eee^{x \cos \theta} \ddd \theta)^2 }\\
 &\hspace{2cm} - \frac{ (\int_\S \cos \theta \eee^{x \cos(\theta)} \ddd \theta)^2}{(\int_\S \eee^{x \cos \theta} \ddd \theta)^2 }\nonumber\\
	&= \frac{\int_\S \cos^2 \theta \eee^{x \cos \theta} \ddd \theta}{\int_\S \eee^{x \cos \theta } \ddd \theta } - V^2(x),\\
	&= \frac{1}{2}\frac{\I_2(x)}{\I_0(x)} - V^2(x). \label{eq:pe3}
\end{align}
The following identity in \cite{Watson1944} is crucial:
\begin{equation}
x \I_2(x) + 2 \I_1(x) - x \I_0(x) = 0, \quad \text{for all } x \in \R\setminus \{0\}. \label{eq:pe1}
\end{equation}
Rewriting (\ref{eq:pe1}), we get
\begin{equation}
\frac{\I_2(x)}{\I_0(x)} = 1 - \frac{2}{x} \frac{\I_1(x)}{\I_0(x)} = 1 -  \frac{2}{x} V(x). \label{eq:pe2}
\end{equation}
Substitution of (\ref{eq:pe2}) into (\ref{eq:pe3}) gives the result for the first derivative of $V$. For the second derivative of $V$, note that
\begin{align}
V''(x) &= \frac{\p}{\p x} \left( - V(x)^2 - \frac{V(x)}{x} + 1 \right)\nonumber\\
 &= - 2 V(x) V'(x) - \frac{V'(x)}{x} + \frac{V(x)}{x^2}. \label{eq:pe4}
\end{align}
The result follows by substitution of $V'(x) = - V(x)^2 - \frac{V(x)}{x} + 1$ into the right-hand side of (\ref{eq:pe4}).
\end{proof}

\begin{lemma}[Properties of $W$]\label{thm:W} The function $W$ has the following properties:
\begin{enumerate}
\item $\lim_{x\to 0} W(x) = 1$,
\item $\lim_{x \to \infty} W(x) = 0$. 
\item $x \mapsto W(x)$ is continuous and strictly decreasing on $(0, \infty)$,
\item $W(-x) = W(x)$.
\end{enumerate}
\end{lemma}

\begin{proof}
The proofs for properties 1--3 are given in \cite[Proposition 3.2]{Meylahn2020}. The last property is an immediate consequence of Property 7 of Lemma \ref{lV}.
\end{proof}

\begin{lemma}[Asymptotics for $V$]\label{thm:Vasym}
The following asymptotics for $V$ can be found:
\begin{align}
V(x) 		&\sim \frac{x}{2}, \quad \text{as } x \downarrow 0, \\
1 - V(x)	&\sim \frac{1}{2x}, \quad \text{as } x \to \infty. 
\end{align}
\end{lemma}

\begin{proof}
The proof can be found in \cite[p.213]{GA18} 
\end{proof}

\begin{lemma}[Bounds for $V$]\label{thm:Vbounds} For all $x \in (0, \infty)$ we have 
\begin{equation}
V^-(x) \leq V(x) \leq V^+(x),
\end{equation}
with 
\begin{align}
V^+(x) 	&= \frac{2x}{1 + 2x},\\
V^-(x)	&= \frac{x}{2 + x}.
\end{align}
\end{lemma}
\begin{proof}
The proof can be found in \cite[Lemma 5, p. 213-214]{GA18}.
\end{proof}

\section{Properties and Asymptotics of r(C)}
\begin{proposition}[Properties of $C \mapsto r(C)$]
\label{prop:r(C)}
If $C > 2$, then the following hold:
\begin{enumerate}
\item $C \mapsto r(C)$ is strictly increasing,
\item $\lim_{C \to \infty} r(C) = 1$,
\end{enumerate}
\end{proposition}

\begin{proof}
Properties $1$ and $2$ follow from properties $3$ and $4$ of Lemma \ref{lV}.
\end{proof}

\begin{proposition}[Asymptotics for $r(C)$]
\label{prop:asympr(C)} The following hold:
\begin{enumerate}
\item $r(C) \sim \sqrt{C-2}$ as $C \downarrow 2$,
\item $1 - r(C) \sim \frac{1}{2C}$ as $C \to \infty$. 
\end{enumerate}

\end{proposition}

\begin{proof}
Note that 
\begin{equation}
\frac{\ddd r(C)}{\ddd C} = \frac{\ddd V(C r(C) )}{\ddd C} = ( r(C) + C r'(C) ) V'(C r(C) ), \label{eq:r'}
\end{equation}
and therefore
\begin{equation}
\frac{r'(C)}{r(C)} = \frac{V'(C r(C) )}{1 - C V'(C r(C))}. \label{eq:sep1}
\end{equation}
Furhtermore, note that
\begin{equation}
V'(C r(C)) \to \frac{1}{2} \label{eq:limV'}
\end{equation}
as $C \downarrow 2$, so that the denominator of \eqref{eq:sep1} tends to zero as $C \downarrow 2$. For this reason, we perform an expansion of $V'(C r(C))$ around $0$ in the denominator. It follows that
\begin{equation}
V'(C r(C)) = \frac{1}{2} - \frac{3}{16} (C r(C))^2 +  O((C r(C))^3). \label{eq:tV'}
\end{equation}
Take $\ee = C - 2$, by using \eqref{eq:limV'} in the numerator of \eqref{eq:sep1} and using \eqref{eq:tV'} in the denominator of \eqref{eq:sep1} we get
\begin{align}
\frac{r'(\ee)}{r(\ee)} \sim \frac{1}{2 - (\ee + 2)( 1- \frac{3}{8} ( \ee + 2)^2 r(\ee)^2 ) },
\end{align}
as $\ee \downarrow 0$.
Neglecting all higher order terms, we get
\begin{equation}
\frac{r'(\ee)}{r(\ee)} \sim \frac{1}{3 r(\ee)^2 - \ee} \label{eq:sep2}
\end{equation}
as $\ee \downarrow 0$. It follows that the ODE 
\begin{equation}
\frac{{r^*}'(\ee)}{r^*(\ee)} = \frac{1}{3 {r^*(\ee)}^2 - \ee}
\end{equation}
has a solution $r^*(\ee) = \sqrt{\ee}$. Hence $r(\ee) \sim r^*(\ee) = \sqrt{\ee}$ as $\ee \downarrow 0$, from which the result follows.

Note that $1 - V(x) \sim \frac{1}{2x}$ as $x \to \infty$ and since $C r(C) \sim C$ as $C \to \infty$, the result follows.

\end{proof}

\begin{remark}
The synchronization level $r(C)$ is approximated well (see Figure \ref{fig:approx}) on the entire domain by
\begin{equation}
r(C) \approx \sqrt{\frac{C-2}{C-1}}.
\end{equation}
\end{remark}

\begin{figure}[!ht]
\begin{center}
\includegraphics[width=0.4\textwidth]{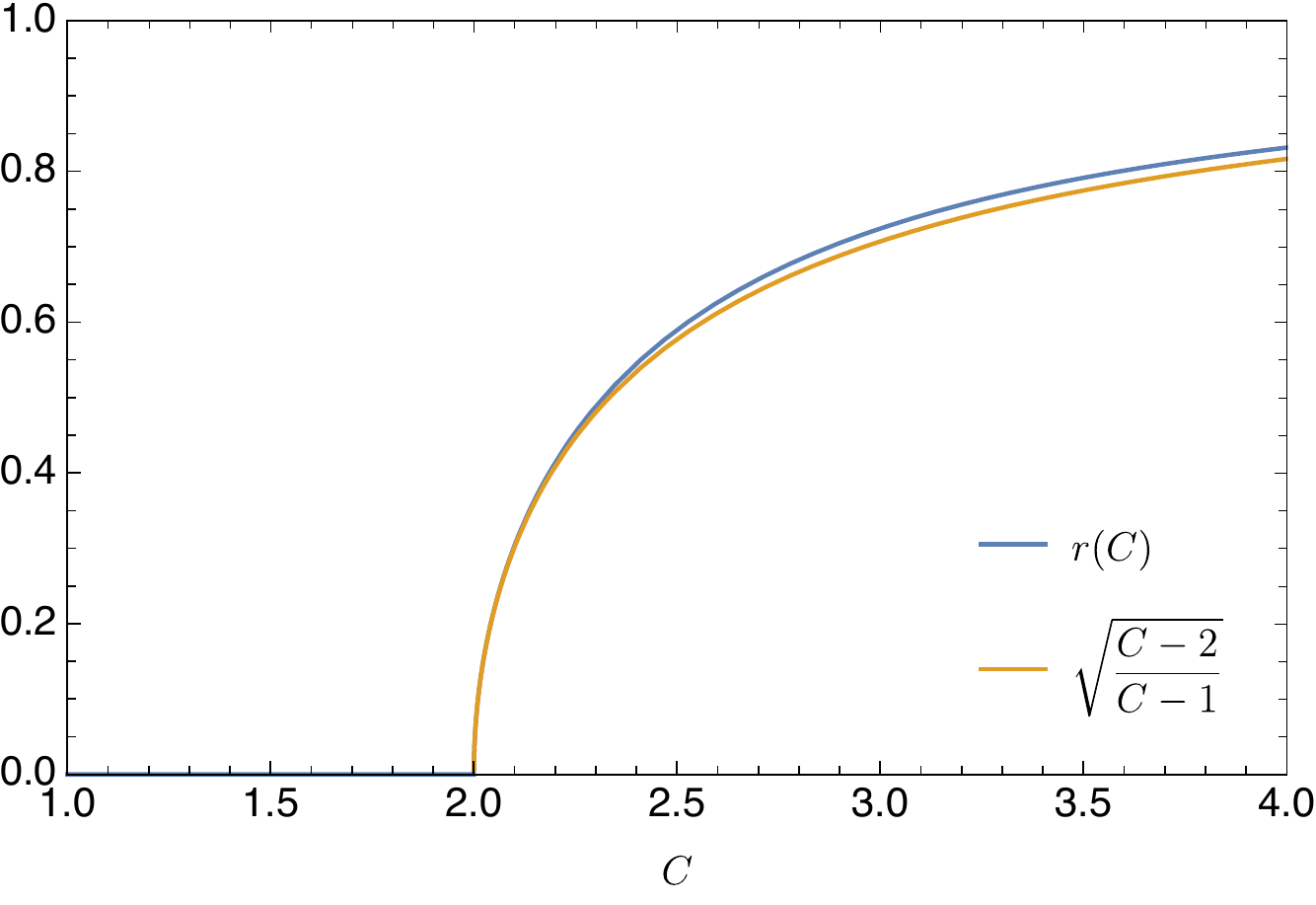} 
\end{center}  
\caption{Plot of $C \mapsto r(C)$ and $C \mapsto \sqrt{\frac{C-2}{C-1}}$.}
\label{fig:approx}
\end{figure}

%\nocite{*}
%\bibliography{aipsamp}% Produces the bibliography via BibTeX.
%%%%%%%%%%% REFERENCES %%%%%%%%%%%%%%%%%%%%%

\end{document}